%% file: PPT-distillable-arXiv-v1.tex
\documentclass[11pt,notitlepage,tightenlines,nofootinbib,superscriptaddress]{revtex4-2}
\input{macros}

\renewcommand{\PPT}{\mathrm{PPT}}
\newcommand{\LL}{\pazocal{L}}
\newcommand{\RR}{\pazocal{R}}
\newcommand{\KK}{\pazocal{K}}
\newcommand{\JJ}{\pazocal{J}}
\newcommand{\GG}{\pazocal{G}}

\newcommand{\pt}{\Gamma}

\makeatletter
\AtBeginDocument{\long\def\@makecaption#1#2{%
  \par\vskip\abovecaptionskip
  \begingroup\normalfont\normalsize\justifying\noindent
  #1: #2\par\endgroup\vskip\belowcaptionskip}}
\makeatother

\newcommand{\Eppt}{E_{d,\,\PPT}}
\newcommand{\EpptErr}{E_{d,\PPT}^{\mathrm{err},\dag}}
\newcommand{\Dbin}{D_{\mathrm{bin}}}

\hypersetup{colorlinks=true,linkcolor=Blues5seq5,citecolor=Blues5seq5,urlcolor=Blues5seq5,
 pdftitle={On PPT entanglement distillation},pdfauthor={Ludovico Lami}}
 
\begin{document}

\title{On PPT entanglement distillation}

\author{Ludovico Lami}
\email{ludovico.lami@gmail.com}
\affiliation{Scuola Normale Superiore, Piazza dei Cavalieri 7, 56126 Pisa, Italy}

\begin{abstract}
We study entanglement distillation under operations that remain completely positive under partial-transpose conjugation, a.k.a.\ PPT channels, in the standard quantum Shannon theory regime of asymptotically vanishing but non-zero error. Our main results are: (a)~a regularised formula for the PPT distillable entanglement in terms of a measured-relative-entropy-like quantity; and (b)~two different single-letter converses, one based on operator quadratic forms and the other on Hirschman's strengthening of the Hadamard three-line theorem. As an immediate application of~(b), we show that the PPT distillable entanglement can be strictly smaller than the regularised Rains bound, thereby resolving an open problem in the theory of entanglement manipulation~\href{https://iopscience.iop.org/article/10.1088/1367-2630/ab4732}{[Regula et al., \emph{NJP} 21:103017, 2019]}. A gap appears already for $3\times 3$ Werner states: at antisymmetric weight $25/26$, a certified upper bound of $0.62107$ ebits lies below the (regularised) Rains bound, which equals $\frac{25}{26} \log_2 5 - \log_2 3 \approx 0.64766$ ebits. On a different note, (a)~implies a faithful lower bound on the PPT distillable entanglement in terms of the entanglement negativity, which provides a quantitative counterpart to the qualitatively known fact that any NPT state is PPT distillable~\href{https://journals.aps.org/prl/abstract/10.1103/PhysRevLett.87.257902}{[Eggeling et al., \emph{PRL} 87:257902, 2001]}.
\end{abstract}

\maketitle

\begingroup\footnotesize\setlength{\baselineskip}{9.5pt}
\makeatletter
\def\l@f@section{\addpenalty{\@secpenalty}\addvspace{0.3em plus 0.5pt}}
\makeatother
\tableofcontents
\endgroup

\section{Introduction}

\noindent \textbf{Note.} The main results and techniques presented here have been derived almost entirely by ChatGPT Astra. The human author (yours truly) gave many suggestions on how to extend and generalise the initial findings, which applied only to Werner states, and spent many hours (but an order of magnitude fewer than would have been needed a year ago) validating and improving the proof techniques, as well as adapting the presentation to his style.

\subsection{Background}

Understanding asymptotic entanglement manipulation is a truly daring endeavour. For how central the phenomenon is for foundations~\cite{Werner, Aspect, dense-coding, teleportation, Ekert91} as well as applications~\cite{teleportation-experiment, entanglement-1200-km, entanglement-optical-clock, DIQKD-experiment, distributed-quantum-computing}, we only have a limited knowledge of the ultimate operational constraints on entanglement transformation in the asymptotic regime, where arbitrarily many copies of a given source state are available.
A particularly challenging task, in this sense, is entanglement distillation~\cite{Bennett-distillation, Bennett-distillation-mixed, Bennett-error-correction, devetak2005}, which consists in turning copies of a noisy, possibly mixed quantum state into as many `ebits' as possible with asymptotically vanishing error. Here, an ebit, also called two-qubit maximally entangled state or Bell pair, can be thought of as the `golden unit' of entanglement, since it is the fuel on which many quantum information processing protocols run. The ultimate rate of ebit production that is achievable from copies of a state $\rho = \rho_{AB}$ through local operations and classical communication (LOCC) is called the LOCC distillable entanglement of $\rho$, and denoted $E_{d,\,\locc}(\rho)$. 

Calculating $E_{d,\,\locc}(\rho)$ for even very simple mixed states $\rho$ is typically an extremely challenging task. As far as the author is aware, it is even unclear whether $E_{d,\,\locc}(\rho)$ is Turing computable at all. Roughly speaking, the nature of this difficulty can be traced back to three different obstacles:
\begin{enumerate}[(1)]
\item The set of LOCC operations is mathematically complicated to describe and analyse~\cite{LOCC}.
\item The presence of the approximation error makes the output of the distillation protocol a moving target.
\item The asymptotic limit of an arbitrarily high number of input copies is analytically intractable, as we have no guarantee on how many copies may be needed to enable distillation~\cite{Watrous2004}.
\end{enumerate}
A possible approach to the problem is to separate these three obstacles, tackling them first one by one, and then in pairs, with the goal of eventually being able to tackle all three at once.

\subsection{Prior work}

A popular idea to remove obstacle~(1) without trivialising the problem is that of replacing LOCC operations with quantum channels that completely preserve the property of having a non-negative partial transpose~\cite{PeresPPT}, also called PPT channels~\cite{Rains1999, Rains2001}. When acting on a bipartite system $A:B$ and outputting a (possibly different) bipartite system $A':B'$, PPT channels can also be defined as those whose Choi state~\cite{Choi} is PPT across the bipartition $AA':BB'$. Therefore, importantly, PPT-ness is a positive semi-definite constraint, entailing that it can be treated efficiently within the formalism of semi-definite programs (SDPs)~\cite{BOYD, vandenberghe_1996}. PPT channels form a (strict) superset of LOCC, meaning that the PPT distillable entanglement of a state is no smaller than its LOCC counterpart: formally, 
\bb
E_{d,\,\locc}(\rho) \leq E_{d,\,\ppt}(\rho)\, .
\ee
Hence, any upper bound on $E_{d,\,\ppt}(\rho)$ automatically yields an upper bound on $E_{d,\,\locc}(\rho)$. It is in this way that one of the most powerful upper bounds on $E_{d,\,\locc}(\rho)$, called the \emph{Rains bound}, was derived~\cite{Rains2001}. This is defined as
\bb
R(\rho) \coloneqq \inf_{\sigma \geq 0\,:\ \|\sigma^{\pt}\|_1\leq 1} D(\rho\|\sigma)\, ;
\ee
here, $D(\rho\|\sigma) \coloneqq \Tr \rho(\log \rho - \log \sigma)$ is the Umegaki relative entropy~\cite{Umegaki1962}, $\Gamma$ denotes the partial transpose on $B$, identified with the linear map satisfying $(X\otimes Y)^{\pt} = X \otimes Y^\intercal$, and $\|Z\|_1\coloneqq \Tr \sqrt{Z^\dag Z}$ denotes the trace norm. Note that the optimisation variable $\sigma$ does not need to be a state\,---\,it can be merely \emph{sub}-normalised, meaning that $\Tr \sigma\leq 1$.\footnote{Sub-normalisation is needed because $1\geq \left\|\sigma^{\pt}\right\|_1 \geq \Tr \sigma^{\pt} = \Tr \sigma$.} The tightest known bound on $E_{d,\,\ppt}$, and one of the tightest known on $E_{d,\,\locc}$, is obtained by \emph{regularising} the Rains bound, i.e.\ by computing its asymptotic value in the limit of arbitrarily many copies:
\bb
E_{d,\,\locc}(\rho) \leq E_{d,\,\ppt}(\rho) \leq R^\infty(\rho) \coloneqq \lim_{n\to\infty} \frac1n\, R\big(\rho^{\otimes n}\big)\, .
\label{Rains_bound}
\ee
For context, the \emph{only} other (upper) bound on $E_{d,\,\locc}$ that, to the best of our knowledge, can sometimes compete with the regularised Rains bound in~\eqref{Rains_bound} is the \emph{squashed entanglement}~\cite{Tucci1999, squashed, rel-ent-sq, Shirokov-sq}. It is known that, in general, the regularisation in~\eqref{Rains_bound} is necessary~\cite{nonadditivity-Rains}. However, unlike many other regularised measures in quantum information theory, $R^\infty$ can actually be shown to be computable, with precise runtime guarantees~\cite{FFF2025}.

Throughout this paper, we will deal with PPT operations instead of the much more complex set of LOCC operations, thereby softening the above obstacle~(1). The key question we face in this work is:
\begin{center}
\it At least for PPT instead of LOCC operations, can we tackle obstacles~(2) and~(3) together?
\end{center}
It is worth remarking that, \emph{if we accept removing also obstacle~(2)}, namely, if we look at exact, zero-error transformations, then much more is known. The task of dilution rather than distillation, in the PPT zero-error case, has been solved completely, first in~\cite{Martin-exact-PPT} for states with vanishing bi-negativity\,---\,in particular, for all two-qubit states~\cite{Ishizaka2004}\,---\,and then in~\cite{PPT-zerocost} for general states. This latter paper built on a previous claimed solution to the problem~\cite{Xin-exact-PPT, Wang2023}, which however turned out to have a serious issue~\cite{Xin-Mark-errata}. See also the recent work~\cite{Brinster2026}. As for distillation, the optimal PPT rate that is achievable with zero error was conjectured to reduce to the min-Rains relative entropy~\cite{Wang2016, nonadditivity-Rains}, a statement that was disproved very recently~\cite{Zhu2026}.

Here, we will leave aside these results, which deal with a setting in which both obstacles~(1) and~(2) have been artificially tamed, and focus on the setting in which only~(1) has been removed artificially (by relaxing LOCC to PPT) but~(2) and~(3) are dealt with together. We will therefore look at the theory of PPT entanglement manipulation in the asymptotically vanishing error setting. In this setting, much less is known from the prior literature. Here is a summary of what we do know:
\begin{itemize}
\item The theory of PPT entanglement manipulation is asymptotically irreversible~\cite{irreversibility-PPT}, something that cannot be amended even by going to the wider class of PPT-preserving operations~\cite{irreversibility}; however, we currently have no expressions for either the PPT distillable entanglement or the PPT entanglement cost.

\item We suspect that the first inequality in~\eqref{Rains_bound} is not tight in general, as \emph{NPT bound entangled states}, i.e.\ states that are not PPT but such that $E_{d,\,\locc}(\rho) = 0$, are conjectured to exist~\cite{HorodeckiBound, Horodecki-open-problems}. Any such state would necessarily satisfy $E_{d,\,\ppt}(\rho) > 0$, because every NPT state is distillable with PPT operations~\cite{Eggeling2001}:
\bb
\rho^\Gamma \not\geq 0\quad\Longrightarrow\quad E_{d,\,\ppt}(\rho) > 0\, .
\label{NPT_states_are_PPT_distillable}
\ee
A remarkable feature of the above result is that it is exquisitely qualitative, in that no lower bound on $E_{d,\,\ppt}(\rho)$ is stated based on the \emph{entanglement negativity}~\cite{negativity} of $\rho$.

\item What about the second inequality in~\eqref{Rains_bound}? 
In~\cite[Section~3.1.1]{Regula2019}, it is explicitly conjectured that it be an equality for all states:
\bb
E_{d,\,\ppt}(\rho) \eqt{?} R^\infty(\rho)\, .
\label{Bartosz_conjecture}
\ee
Solving this conjecture, or otherwise finding an alternative formula for $E_{d,\,\ppt}(\rho)$, would mark a decisive step forward in our understanding of PPT operations.
\end{itemize}

\subsection{Our contributions}

In this paper, we investigate in a systematic way entanglement distillation under PPT operations in the regime of asymptotically vanishing error. Our main results are as follows:
\begin{enumerate}[(a)]

\item In Theorem~\ref{regularised_formula_thm}, we establish a regularised formula for the PPT distillable entanglement $E_{d,\,\ppt}$. This formula can be used to calculate achievable rates (i.e.\ lower bounds on $E_{d,\,\ppt}$) systematically, but we have not yet been able to provide any guarantees on the speed of convergence. As a byproduct, it yields an explicit faithful lower bound in terms of the entanglement negativity: writing $N(\rho)\coloneqq(\|\rho^\Gamma\|_1-1)/2$ and $d\coloneqq\min\{d_A,d_B\}$, we obtain $E_{d,\,\ppt}(\rho)\geq1-h_2\!\left(\tfrac12+\tfrac{N(\rho)}{d+1}\right)$. This makes the qualitative asymptotic statement~\eqref{NPT_states_are_PPT_distillable} quantitative by supplying a strictly positive rate whenever $N(\rho)>0$.

\item Then, in Theorems~\ref{quadratic_converse_thm} and~\ref{analytic_converse_thm} we prove two different converse bounds, i.e.\ two upper bounds on $E_{d,\,\ppt}$ that are valid for any state. These are quite different in nature: the one in Theorem~\ref{quadratic_converse_thm} is based on quadratic forms, while that in Theorem~\ref{analytic_converse_thm} is based on the theory of analytic functions.

Since every LOCC operation is PPT, these converses also upgrade the best available (computable) upper bounds on the LOCC distillable entanglement to
\bb
E_{d,\,\locc}(\rho) \leq \min\big\{ R^\infty(\rho),\, E_{\mathrm{sq}}(\rho),\, E_Q(\rho),\, E_A(\rho)\big\}\, .
\label{LOCC_converse_bounds}
\ee
Here, $E_{\mathrm{sq}}$ denotes the squashed entanglement~\cite{Tucci1999, squashed, rel-ent-sq, Shirokov-sq}. The last two bounds are new, and are defined in~\eqref{quadratic_bound} and~\eqref{analytic_bound}, respectively.

\item A remarkable consequence of our converses is that the conjectured equality~\eqref{Bartosz_conjecture} between the PPT distillable entanglement and the regularised Rains bound is, alas, false: in Theorem~\ref{gap_Werner_thm}, we find that, already for a $3\times 3$ Werner state with antisymmetric weight $25/26$, we have $\Eppt(\rho_{25/26}) \leq 0.62107 < \frac{25}{26}\log5 - \log3 = R^\infty(\rho_{25/26})$.

\end{enumerate}

\section{Main results}
\label{main_results_sec}

All systems are finite dimensional and logarithms are base two. We write $\Gamma$ for the partial transpose on the second party in a fixed product basis. A state is PPT if $\rho^{\pt} \geq 0$, and a channel $\Lambda$ is PPT if $\Gamma_{\rm out}\circ\Lambda\circ\Gamma_{\rm in}$ is completely positive. The PPT distillable entanglement $\Eppt(\rho)$ is the asymptotic number of ebits per copy obtainable with vanishing error. We also consider its strong-converse threshold $\Eppt^\dag(\rho)$, above which the optimal overlap with the target maximally entangled state tends to zero, and the lower exponential decay rate $\EpptErr(\rho,R)$ of that overlap at rate $R$. Precise operational definitions are given in Section~\ref{setting_sec}.

For $\Phi_d = \frac1d \sum_{i,j = 1}^d \ketbraa{ii}{jj}$, the normalised maximally entangled state, write
\bb
F_{\PPT}(\rho,d) \coloneqq \sup_{\Lambda\in\PPT}\Tr\big[\Phi_d\Lambda(\rho)\big]\, ,
\label{bell_overlap_definition}
\ee
where the supremum is over PPT channels with the indicated output spaces. At rate $R \geq 0$, the exponent is
\bb
\EpptErr(\rho,R) \coloneqq \liminf_{n\to\infty} - \frac1n\log F_{\PPT}\!\left(\rho^{\otimes n},\left\lceil 2^{nR}\right\rceil\right), \qquad R \geq 0\, .
\label{distillation_error_exponent}
\ee
A positive exponent gives an exponential strong converse.

\subsection{A regularised formula}

Our first result expresses the distillation rate through a measurement optimisation. Let $\MM$ be a finite measurement channel represented by the POVM $(E_j)_j$, and define
\bb
L(\rho) \coloneqq \sup_{\MM}\sum_{j:\ \Tr[\rho E_j]>0}\Tr[\rho E_j]\log\frac{\Tr[\rho E_j]}{\big\|E_j^{\pt}\big\|_\infty}\, ,\qquad L^\infty(\rho) \coloneqq \lim_{n\to\infty}\frac1n L\big(\rho^{\otimes n}\big)\, .
\label{L_definition}
\ee
Zero effects are omitted. The limit exists by superadditivity and the logarithmic-negativity upper bound proved in Proposition~\ref{FG_properties_prop}. Unlike a finite-error operational optimisation, $L$ contains neither an output dimension nor a fidelity constraint.

Rains's projective-measurement bound~\cite[Theorem~6.4]{Rains2001} is recovered by restricting the POVM to orthogonal projections. Allowing arbitrary POVMs and then regularising gives an equality with the operational rate. The converse direction uses the two-outcome measurement obtained by pulling back the Bell test of a distillation protocol; the direct direction combines a fixed measurement with classical typicality. Thus the same explicit functional accounts for both sides of the rate formula.

\begin{thm}[(Regularised formula for the PPT distillable entanglement)] \label{regularised_formula_thm}
For every finite-dimensional bipartite state, with $L$ defined by~\eqref{L_definition},
\bb
\Eppt(\rho) = L^\infty(\rho) = \sup_{n\in\N^+}\frac1n\,L\big(\rho^{\otimes n}\big)\, .
\label{regularisation}
\ee
\end{thm}

For a state with computable entries, Theorem~\ref{regularised_formula_thm} gives a computable increasing sequence of lower bounds converging to $\Eppt(\rho)$: finite rational POVMs and positive rational weights can be enumerated, and the corresponding scores can be approximated from below. Thus $\Eppt$ is lower semicomputable. This also follows directly by searching over finite-block PPT channels and applying the hashing inequality to their outputs~\cite{devetak2005}; no new computability status is claimed here. Had $\Eppt$ equalled $R^\infty$, the algorithms of~\cite{FFF2025} would have established its Turing computability. Our separation rules out that identification, but does not rule out computability by another method; the existence of a convergent computable family of upper bounds remains open.

For the proofs and sharper finite-block lower bounds, we retain an auxiliary quantity $\widetilde{L}$, defined in~\eqref{povm_regularised_scores}, whose regularisation $\widetilde{L}^{\,\infty}$ also equals $\Eppt$. It keeps the operator absolute values before taking a norm. This additional structure gives simple spectral lower bounds and the one-copy Werner calculation, but is unnecessary in the statement of Theorem~\ref{regularised_formula_thm}.

\subsection{Two converse constructions}

Our two upper bounds use different information about a distillation test. The first controls a quadratic form on the test, including correlations with an ancillary system. The second places the input in an analytic family and combines bounds on both boundary circles of an annulus. Both retain enough finite-block information to prove exponential strong converses, beyond an upper bound on vanishing-error rates.

For the quadratic construction, the operator space carries the Hilbert--Schmidt scalar product $\Braket{X,Y} = \Tr[X^\dag Y]$. Write $\LL_\omega(X) \coloneqq \omega X$. For $\alpha>0$ and a strictly positive super-operator $\KK$ on this space, define
\bb
V_\alpha[\omega,\KK] \coloneqq \frac1\alpha\Braket{\sqrt\omega,\,\big(\log\big(\LL_{\sqrt\omega}\circ\KK^{-1}\circ\LL_{\sqrt\omega}\big)\big)(\sqrt\omega)}\, .
\label{V_definition}
\ee
Here $\omega \geq 0$ may be singular, and the logarithm is restricted to the support of the super-operator. The vector $\sqrt\omega$ lies in that support. Positivity of super-operators refers to the Hilbert--Schmidt scalar product, not to complete positivity.

\begin{thm}[(Conditional quadratic converse)]\label{quadratic_converse_thm}
Let $\rho = \rho_{A:B}$ be a bipartite state. With $V_\alpha$ defined by~\eqref{V_definition}, put
\bb
E_Q(\rho) = E_Q(\rho_{A:B}) \coloneqq \inf_{\rho_{AA':BB'},\,\alpha,\KK,\JJ} \left\{V_\alpha[\rho_{AA':BB'},\KK] - V_\alpha[\rho_{A':B'},\JJ]\right\} ,
\label{quadratic_bound}
\ee
where the infimum is over all finite-dimensional extensions $\rho_{AA':BB'}$ of $\rho_{A:B}$, all $\alpha>0$, and all strictly positive super-operators $\KK,\JJ$ on the operator spaces of $AA'BB'$ and $A'B'$, respectively, satisfying
\bb
\Braket{Q_n\otimes X_n,\KK^{\otimes n}[Q_n\otimes X_n]} \leq \big\|Q_n^{\Gamma_{B^n}}\big\|_\infty^\alpha\Braket{X_n,\JJ^{\otimes n}[X_n]}
\label{ancillary_form_all_copies}
\ee
for all $n\in\N^+$, all tests $0 \leq Q_n \leq \id$ on $A^nB^n$, and all operators $X_n$ on $(A'B')^n$. Then
\bb
\Eppt(\rho_{A:B}) \leq \Eppt^\dag(\rho_{A:B}) \leq E_Q(\rho_{A:B})\, .
\label{conditional}
\ee
Every rate strictly above $E_Q(\rho_{A:B})$ has an exponential strong converse. 
\end{thm}

The all-copy hypothesis is a substantive constraint; it cannot in general be replaced by its one-copy version. Geometric means of positive super-operators turn this quadratic control into data processing for $V_\alpha$. Conditioning on an ancillary system then subtracts the corresponding ancillary contribution. The choice of right multiplication recovers the Rains bound, while different quadratic forms can improve it, as the Werner example below shows.

For positive operators, we write
\bb
D(X\|Y) \coloneqq \Tr\big[X(\log X - \log Y)\big] \quad\text{if }\supp X \subseteq \supp Y\, ,
\label{relative_entropy_definition}
\ee
and set $D(X\|Y) = +\infty$ otherwise. Our relative entropy for a subnormalised second argument has no trace correction.

The sandwiched R\'enyi relative entropy~\cite{tomamichel12smooth_tutorial, Wilde2014, newRenyi} is
\bb
\widetilde D_\alpha(\rho\|\sigma) \coloneqq \frac1{\alpha - 1}\log\Tr\!\left[ \left(\sigma^{\frac{1 - \alpha}{2\alpha}}\rho \sigma^{\frac{1 - \alpha}{2\alpha}}\right)^\alpha\right], \qquad \alpha>1\, ,
\label{sandwiched_renyi}
\ee
when $\supp\rho\subseteq\supp\sigma$, and is $+\infty$ otherwise. Negative powers are restricted to $\supp\sigma$.

The analytic converse uses an explicit positive weight on the two boundary circles. For $0<\theta<1$, put
\bb
\beta_\theta(s) \coloneqq \frac{\sin(\pi\theta)}{2\theta\bigl(\cosh(\pi s) + \cos(\pi\theta)\bigr)},\qquad s\in\mathbb R\, .
\label{hirschman_density}
\ee
For an annulus $\mathcal{O} = \{z:a<|z-c|<b\}$ and an interior point $z_* = c + r e^{i\varphi}$, write $\ell \coloneqq \ln(b/a)$ and $\theta \coloneqq \ln(r/a)/\ell$. Define the boundary density
\bb
k_{\mathcal{O},z_*}\!\left(c + a e^{i\psi}\right) &\coloneqq \frac{1 - \theta}{a\ell}\sum_{m\in\mathbb Z}\beta_{1 - \theta}\!\left(\frac{\psi - \varphi + 2\pi m}{\ell}\right),\\
k_{\mathcal{O},z_*}\!\left(c + b e^{i\psi}\right) &\coloneqq \frac{\theta}{b\ell}\sum_{m\in\mathbb Z}\beta_\theta\!\left(\frac{\psi - \varphi + 2\pi m}{\ell}\right)\, .
\label{explicit_annular_density}
\ee
The series converge uniformly and define positive $2\pi$-periodic functions. With respect to positive arclength $|\mathrm dz|$, the inner and outer weights have masses $1 - \theta$ and $\theta$, respectively. Their sum is one, and no orientation of either circle is involved.

\begin{thm}[(Analytic-family strong converse)]\label{analytic_converse_thm}
Let $\rho = \rho_{A:B}$ be a finite-dimensional bipartite state, let $\mathcal{O} = \{z\in\mathbb C:a<|z - c|<b\}$ with $c\in\mathbb C$ and $0<a<b$, and let $Z(z)$ be an operator on $AB$, analytic in a neighbourhood of $\overline{\mathcal{O}}$, with $\Tr[Z(z)] = 1$. Suppose that $z_*\in\mathcal{O}$, $\sigma = Z(z_*)$ is a state, and $\supp\rho\subseteq\supp\sigma$. With the density in~\eqref{explicit_annular_density}, built from~\eqref{hirschman_density}, define
\bb
J_{Z,\mathcal{O},z_*}(R) \coloneqq \int_{\partial\mathcal{O}} \min\!\left\{\log\|Z(z)\|_1,\,\log\big\|Z(z)^{\pt_B}\big\|_1 - R\right\}k_{\mathcal{O},z_*}(z)\,|\mathrm dz|\, .
\label{analytic_J}
\ee
If $R \geq 0$ and $D(\rho\|\sigma) + J_{Z,\mathcal{O},z_*}(R)<0$, then $R$ is an exponential strong-converse rate for PPT distillation of $\rho$. Specifically,
\bb
\EpptErr(\rho,R) \geq \sup_{\alpha>1}\frac{\alpha - 1}{\alpha}\Big[ - J_{Z,\mathcal{O},z_*}(R) - \widetilde D_\alpha(\rho\|\sigma)\Big]>0\, .
\label{analytic_error_exponent}
\ee
\end{thm}

The analytic family need not consist of states away from the reference point. This freedom is useful: trace-norm duality supplies two different boundary bounds, and Hirschman's inequality averages their smaller value with positive weights. Quantum hypothesis testing then transfers the resulting exponential bound from the reference state to the input. For a fixed reference family this yields an explicit scalar boundary integral, while optimisation over families gives the converse $E_A$ defined in~\eqref{analytic_bound}.

\subsection{A strict gap from the regularised Rains bound}

For local dimension $d$, let $F_d$ be the swap, $P_0 = (\id + F_d)/2$ and $P_1 = (\id - F_d)/2$. The Werner states are
\bb
\rho_p \coloneqq (1 - p)\frac{2P_0}{d(d + 1)} + p\frac{2P_1}{d(d - 1)}\, ,\qquad 0 \leq p \leq 1\, ,\quad d \geq 2\, .
\label{werner}
\ee
The regularised Rains bound $R^\infty$, defined in~\eqref{rains_definition}, is a standard upper bound on $\Eppt$. It was not known whether this bound always gives the PPT distillation rate~\cite{Regula2019}. The following separation shows that it does not, already for qutrit Werner states.

Symmetry makes the relevant calculations small enough to certify. On the lower-bound side, Werner twirling converts collective POVMs into a finite linear program, and one feasible $48$-copy POVM certifies the displayed rate. For the converse, the symmetric and antisymmetric sectors reduce a quadratic-form comparison to two dimensions. An analytic Werner family gives a stronger numerical upper bound; its boundary integral is bounded with directed rounding, rather than inferred from a floating-point optimisation. Figure~\ref{werner_fig} shows the resulting bounds across the NPT interval.

\begin{thm}[(Strict qutrit Werner separation)]\label{gap_Werner_thm}
For $d = 3$ and $p = 25/26$,
\bb
0.5836& \leq \Eppt(\rho_{25/26}) \leq 0.6211\\
&<R^\infty(\rho_{25/26}) = \frac{25}{26}\log5 - \log3 \approx 0.6477\, .
\label{benchmark}
\ee
Each of $E_Q$, $U$ and $E_A$, defined in~\eqref{quadratic_bound}, \eqref{one_state_quadratic_bound} and~\eqref{analytic_bound}, respectively, is strictly below $R^\infty$ on this state.
\end{thm}

\begin{figure}[t]
{\centering\includegraphics[width=.82\linewidth]{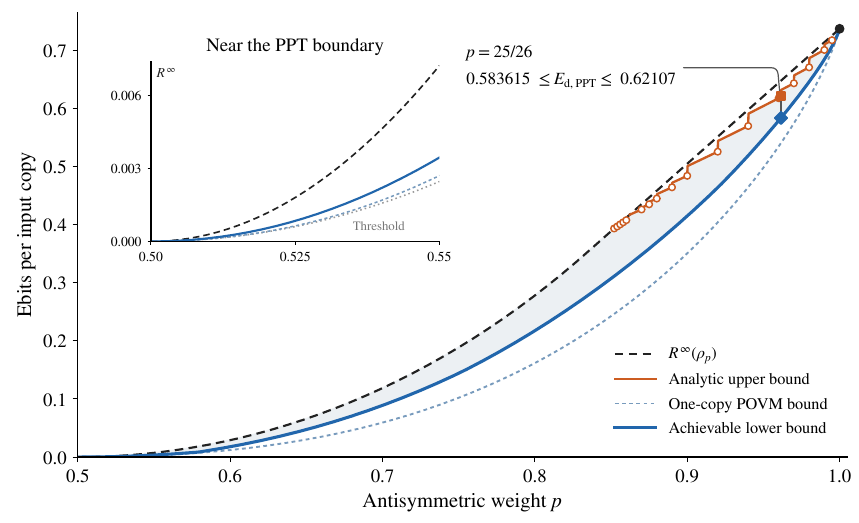}\par}
\justifying
\caption{Bounds for qutrit Werner states; the PPT region $p\leq1/2$ has zero rate. Dashed: regularised Rains bound~\eqref{rains}. Light blue: the one-copy bound $\widetilde{L}(\rho_p)$ in~\eqref{onecopy_value}. Blue: achievable bounds from POVMs on up to $48$ copies and collective PPT effects, including the valid propagation in~\eqref{admixture}. Orange: certified analytic upper bounds with the same propagation rule; the small jumps reflect the finite certificate set. The inset compares the lower bounds near the PPT boundary, including the grey threshold bound~\eqref{werner_lower}. Curves are floating-point evaluations of proved formulas; marked benchmark values have interval certificates. The black endpoint is the known exact antisymmetric rate.}\label{werner_fig}
\end{figure}

\section{Setting and operational reduction}
\label{setting_sec}

We retain the conventions of Section~\ref{main_results_sec}. On $n$ copies, partial transpose acts on the entire second party of the bipartition $A^n:B^n$. Channels are completely positive and trace preserving, and $\ln$ denotes the natural logarithm.

\subsection{One-shot and asymptotic quantities}

For $\e \in (0,1)$, define the one-shot PPT distillable entanglement by the hypothesis-testing expression
\bb
\Eppt^{(1),\,\e}(\rho) \coloneqq - \log \inf_{\substack{0 \leq Q \leq \id\\
\Tr \rho Q \geq 1 - \e}} \big\|Q^{\pt}\big\|_\infty\, .
\label{one_shot_ppt}
\ee
Rigorously speaking, the largest achievable integer Schmidt rank at error $\e$ is $\lfloor2^{\Eppt^{(1),\,\e}(\rho)}\rfloor$, by Rains's semidefinite program~\cite[Theorem~3.1]{Rains2001}. Taking its logarithm changes~\eqref{one_shot_ppt} by less than one bit, which is asymptotically irrelevant.
The PPT distillable entanglement and its strong-converse threshold are, respectively,
\bb
\Eppt(\rho) & \coloneqq \lim_{\e \to 0^+}\liminf_{n\to\infty}\frac1n\Eppt^{(1),\,\e}(\rho^{\otimes n})\, ,\\
\Eppt^\dag(\rho) & \coloneqq \lim_{\e \to 1^-}\limsup_{n\to\infty}\frac1n\Eppt^{(1),\,\e}(\rho^{\otimes n})\, .
\label{operational_rate}
\ee
Thus the first quantity describes rates achievable with vanishing error, whereas every rate above the second forces the best Bell overlap to tend to zero.

For an integer $d \geq 1$, the optimal Bell overlap in~\eqref{bell_overlap_definition} has the semidefinite representation
\bb
F_{\PPT}(\rho,d) = \max\left\{\Tr \rho Q:\ 0 \leq Q \leq \id,\ \big\|Q^{\pt}\big\|_\infty \leq d^{-1}\right\} .
\label{optimal_overlap}
\ee
The exponent in~\eqref{distillation_error_exponent} measures decay of this overlap, rather than of its square root.

For completeness, we recall the implementation of~\eqref{optimal_overlap}. A PPT channel gives $Q = \Lambda^\dag(\Phi_d)$ satisfying
\bb
0 \leq Q \leq \id,\qquad \big\|Q^{\pt}\big\|_\infty \leq d^{-1}\, ,
\label{effect}
\ee
because its conjugate adjoint is unital positive. Conversely, for $d \geq 2$, such a test defines the channel
\bb
\Lambda_Q(X) = \Tr[XQ]\Phi_d + \Tr\big[X(\id - Q)\big]\frac{\id - \Phi_d}{d^2 - 1}\, .
\label{channel}
\ee
Writing $P_\pm = (\id \pm F_d)/2$, where $F_d$ is the target swap, its transpose conjugate is
\bb
(\Gamma_{\rm out} \circ \Lambda_Q \circ \Gamma_{\rm in})(X) = \frac{\Tr\big[X(\id + dQ^{\pt})\big]}{d(d + 1)}P_+ + \frac{\Tr\big[X(\id - dQ^{\pt})\big]}{d(d - 1)}P_-\, .
\ee
Both coefficient matrices are positive, proving complete positivity. At $d = 1$ the optimal overlap is trivially one.

For binary probabilities we use $h_2(a) = - a\log a - (1 - a)\log(1 - a)$ and $\Dbin(a\|b) = a\log(a/b) + (1 - a)\log\big((1 - a)/(1 - b)\big)$, with the usual continuous endpoint conventions. 

\subsection{The Rains bound}
\label{rains_prelim_sec}

The Rains set, its relative entropy, and the corresponding regularisation are~\cite{Rains2001}
\bb
\mathcal{R}(A:B) & \coloneqq \left\{\sigma_{AB} \geq 0:\ \big\|\sigma_{AB}^{\pt_B}\big\|_1 \leq 1\right\} ,\\
R(\rho) & \coloneqq \inf_{\sigma\in\mathcal{R}(A:B)}D(\rho\|\sigma), \qquad R^\infty(\rho) \coloneqq \lim_{n\to\infty}\frac1n R(\rho^{\otimes n})\, .
\label{rains_definition}
\ee
Tensor closure of $\mathcal{R}$ makes $R$ subadditive, so the limit is also the infimum over $n$. Rains's converse gives $\Eppt(\rho) \leq R^\infty(\rho)$. Regularisation cannot in general be omitted: Wang and Duan exhibited two-qubit states with $R(\rho^{\otimes2}) < 2R(\rho)$~\cite[Theorem~1]{nonadditivity-Rains}.

Recall the measured relative entropy~\cite{Donald1986} and the measured Rains bound. A finite measurement channel $\MM$ is represented by a finite POVM $\pazocal{E} = (E_j)_j$ through $\MM(X) = \sumno_j\Tr\!\left[XE_j\right]\ketbraa{j}{j}$, where the output basis is indexed by the outcomes. Then
\bb
D^M(\rho\|\sigma) & \coloneqq \sup_{\MM}D(\MM(\rho)\|\MM(\sigma)), \qquad R^M(\rho) \coloneqq \inf_{\sigma\in\mathcal{R}(A:B)}D^M(\rho\|\sigma)\, ,
\label{measured_rains_definition}
\ee
where the supremum ranges over finite measurement channels. The following useful form follows from the variational and approximation results of Fang, Fawzi and Fawzi~\cite{Fang2025, FFF2025}.

\begin{lemma}\label{measured_rains_lemma}
For every bipartite state $\rho = \rho_{AB}$,
\bb
R^M(\rho) = \sup_{H>0:\,\left\|H^{\pt}\right\|_\infty \leq 1}\Tr\rho\log H\, .
\label{measured_rains_norm}
\ee
Moreover, denoting with $D$ the global dimension of $AB$, for all $n\in\N^+$ we have
\bb
0 \leq R(\rho^{\otimes n}) - R^M(\rho^{\otimes n}) \leq 2(D^2 + D)\log(n + D)\, .
\label{rains_measured_penalty}
\ee
In particular, $R^\infty(\rho) = \lim_{n\to\infty}\frac1n R^M(\rho^{\otimes n})$.
\end{lemma}
\begin{proof}
Combining~\cite[Lemma~19]{Fang2025} with the Rains support-function identity in~\cite[Eq.~(98)]{FFF2025} gives
\bb
R^M(\rho) & = \sup_{H>0:\,h_{\mathcal{R}}(H) \leq 1}\Tr\rho\log H,\qquad h_{\mathcal{R}}(H) \coloneqq \sup_{\sigma\in\mathcal{R}}\Tr\sigma H = \min_{K \geq H}\big\|K^{\pt}\big\|_\infty\, .
\label{rains_support_formula}
\ee
Thus every feasible $H$ admits $K \geq H$ with $\big\|K^{\pt}\big\|_\infty \leq 1$; operator monotonicity of $\log$ allows us to replace $H$ by $K$ without decreasing the objective. Conversely, $\big\|H^{\pt}\big\|_\infty \leq 1$ implies $h_{\mathcal{R}}(H) \leq 1$. This proves~\eqref{measured_rains_norm}. The estimate~\eqref{rains_measured_penalty} is~\cite[Lemma~16, Eq.~(24)]{FFF2025}, applied to $\{\rho^{\otimes n}\}$ and $\mathcal{R}(A^n:B^n)$; its hypotheses are verified for the Rains set in the same reference, Sec.~4.3.
\end{proof}

The following strong superadditivity of the regularised Rains bound was established by Ao, Philip and Streltsov~\cite[Lemma~1]{Ao2026}. We include a simple proof for completeness.

\begin{prop}[{\cite{Ao2026}}]\label{rains_strong_superadditivity_prop}
The regularised Rains bound is strongly superadditive:
\bb
R^\infty(\rho_{AA':BB'}) \geq R^\infty(\rho_{A:B}) + R^\infty(\rho_{A':B'})\, .
\label{rains_strong_superadditivity}
\ee
In particular, $R^\infty(\rho\otimes\sigma) = R^\infty(\rho) + R^\infty(\sigma)$.
\end{prop}
\begin{proof}
Taking product witnesses in~\eqref{measured_rains_norm}, as in the proof of~\cite[Lemma~23 and Appendix~C]{Fang2025}, gives $R^M(\rho_{AA':BB'}) \geq R^M(\rho_{A:B}) + R^M(\rho_{A':B'})$: the logarithm of a tensor product is a sum of local logarithms, even when the input is correlated. Apply this to $\rho_{AA':BB'}^{\otimes n}$, divide by $n$, and use~\eqref{rains_measured_penalty}. Product inputs satisfy the reverse inequality by subadditivity of $R^\infty$.
\end{proof}

Proposition~\ref{rains_strong_superadditivity_prop} concerns arbitrary correlated inputs, rather than only tensor products. Full additivity was proved by Beigi, Rubboli and Tomamichel~\cite[Proposition~9]{BRT26}, using the asymptotic equipartition property of~\cite{Fang2025}. The measured-relative-entropy argument above makes the correlated-input extension immediate.

\section{A regularised formula for PPT distillable entanglement}
\label{achievability_sec}

We construct two functions of a state and a measurement, optimise over the measurement, and then regularise over copies. The first function retains the operator absolute values of the partially transposed effects, while the second replaces each of them by its operator norm. Although the two optimisations can differ on a single copy, their regularisations coincide and give the PPT distillable entanglement, as established in Theorem~\ref{regularised_formula_thm}.

\subsection{Definitions and elementary properties}

Let $\MM$ be a finite measurement channel represented by the POVM $\pazocal{E} = (E_j)_j$, with zero effects omitted. For a bipartite state $\rho$, define
\bb
\widetilde{\ell}(\rho,\MM) \coloneqq \sup_{y>0}\left\{\sumno_j\Tr[\rho E_j]\log y_j - \log\left\|\sumno_j y_j\big|E_j^{\pt}\big|\right\|_\infty\right\} .
\label{povm_f}
\ee
Here $y = (y_j)_j$ is a family of strictly positive real numbers, and $|H| \coloneqq (H^2)^{1/2}$ for a Hermitian operator $H$. The operators $E_j^{\pt}$ need not commute, and their absolute values need not sum to $\id$.

\begin{lemma}[(Dual representation)]\label{povm_dual_lemma}
For every state $\rho$ and finite measurement channel $\MM$ as above,
\bb
\widetilde{\ell}(\rho,\MM) = \min_\sigma \sum_{j:\ \Tr[\rho E_j]>0}\Tr[\rho E_j]\log\frac{\Tr[\rho E_j]}{\Tr\!\left[\sigma\big|E_j^{\pt}\big|\right]}\, ,
\label{povm_dual}
\ee
where the minimum is over normalised states, finite, and attained. Positive-probability terms with zero denominator have value $+\infty$, and zero-probability terms are zero.
\end{lemma}
\begin{proof}
We use the linearisation of the logarithmic penalty in~\cite[Lemma~1]{Berta2017}, followed by Sion's minimax theorem~\cite[p.~174]{Sion}. Optimising a common positive scale of $y$ in~\eqref{povm_f} gives
\bb
\widetilde{\ell}(\rho,\MM) &\eqt{(i)} \sup_{y>0}\min_\sigma\left\{\sumno_j\Tr[\rho E_j]\log y_j + \frac{1 - \sumno_j y_j\Tr\!\left[\sigma\big|E_j^{\pt}\big|\right]}{\ln2}\right\}\\
&\eqt{(ii)} \min_\sigma\sup_{y>0}\left\{\sumno_j\Tr[\rho E_j]\log y_j + \frac{1 - \sumno_j y_j\Tr\!\left[\sigma\big|E_j^{\pt}\big|\right]}{\ln2}\right\} .
\label{povm_sion}
\ee
For (i), the best scale makes $\big\|\sumno_j y_j\,\big|E_j^{\pt}\big|\big\|_\infty = 1$, and the norm is the maximum expectation over states. For (ii), the state space is compact and convex, the positive orthant is convex, and the displayed objective is continuous, affine in $\sigma$ and concave in $y$. These are precisely the hypotheses of Sion's theorem; compactness of the positive orthant is not required. The inner supremum separates over $j$ and gives~\eqref{povm_dual}: when both traces are positive, its maximiser is $y_j = \Tr[\rho E_j]\big/\Tr\!\big[\sigma|E_j^{\pt}|\big]$; the remaining cases follow by taking $y_j \to 0^+$ or $y_j\to\infty$. The maximally mixed state gives positive denominators for every nonzero effect, so the minimum is finite; lower semicontinuity and compactness show that it is attained.
\end{proof}

Bounding each denominator in~\eqref{povm_dual} with $\Tr \sigma\big|E_j^{\pt}\big| \leq \big\|E_j^{\pt}\big\|_\infty$ gives the simpler function
\bb
\widetilde{\ell}(\rho,\MM) \geq \ell(\rho,\MM) \coloneqq \sum_{j:\ \Tr[\rho E_j]>0}\Tr[\rho E_j]\log\frac{\Tr[\rho E_j]}{\big\|E_j^{\pt}\big\|_\infty}\, .
\label{povm_g}
\ee
For a partition of the identity into orthogonal projections on $V\otimes V$, this is the score in~\cite[Theorem~6.4]{Rains2001}. The two scores can be negative. Equality in~\eqref{povm_g} holds exactly when the maximal-eigenvalue eigenspaces of all positive-probability $\big|E_j^{\pt}\big|$ have a nonzero common intersection, since a minimising state in~\eqref{povm_dual} must then attain every norm simultaneously.

For comparison with $L$ in~\eqref{L_definition}, we optimise over finite measurement channels and define the auxiliary functional
\bb
\widetilde{L}(\rho) \coloneqq \sup_{\MM}\widetilde{\ell}(\rho,\MM)\, ,\qquad L(\rho) = \sup_{\MM}\ell(\rho,\MM)\, ,
\label{povm_regularised_scores}
\ee
together with its regularisation, alongside the previously defined $L^\infty$
\bb
\widetilde{L}^{\,\infty}(\rho) \coloneqq \lim_{n\to\infty}\frac1n \widetilde{L}\big(\rho^{\otimes n}\big)\, ,\qquad L^\infty(\rho) = \lim_{n\to\infty}\frac1n L\big(\rho^{\otimes n}\big)\, .
\label{FG_regularisations}
\ee
Both limits are well defined by Proposition~\ref{FG_properties_prop} below. No positive part is needed in these optimisations, since the trivial measurement gives $\widetilde{\ell} = \ell = 0$.

\begin{prop}[(Elementary properties of $\widetilde{L}$ and $L$)]\label{FG_properties_prop}
For finite-dimensional bipartite states, the following hold.
\begin{enumerate}[(i)]
\item The functions $\widetilde{L}$ and $L$ are convex and satisfy
\bb
0 \leq L(\rho) \leq \widetilde{L}(\rho) \leq \log\big\|\rho^{\pt}\big\|_1\, .
\label{fg_bounds}
\ee
In particular, they vanish on PPT states.
\item Both functions are strongly superadditive:
\bb
\widetilde{L}(\rho_{AA':BB'}) & \geq \widetilde{L}(\rho_{A:B}) + \widetilde{L}(\rho_{A':B'})\, ,\\
L(\rho_{AA':BB'}) & \geq L(\rho_{A:B}) + L(\rho_{A':B'})\, .
\label{fg_strong_superadditivity}
\ee
\item For every PPT channel $\Lambda$,
\bb
L\big(\Lambda(\rho)\big) \leq L(\rho)\, .
\label{g_data_processing}
\ee
\item If $0 \leq A \leq \id$, $0 \leq A^{\pt} \leq b\id$, and $0<b<a \coloneqq \Tr[\rho A] \leq 1$, then, for the binary measurement channel $\MM_A$ represented by $(A,\id-A)$,
\bb
\widetilde{L}(\rho) \geq \widetilde{\ell}(\rho,\MM_A) \geq \Dbin(a\|b)\, .
\label{f_binary}
\ee
\item Let $d \coloneqq \min\{d_A,d_B\} \geq 2$, and write $-\nu_1,\ldots,-\nu_r$ for the strictly negative eigenvalues of $\rho^{\pt}$. With
\bb
N \coloneqq \sum_{i = 1}^r\nu_i = \frac{\big\|\rho^{\pt}\big\|_1 - 1}{2},\qquad s \coloneqq \left(\sumno_{i = 1}^r\nu_i^2\right)^{1/2}\, ,
\label{npt_spectral_parameters}
\ee
we have
\bb
\widetilde{L}(\rho) &\geq \max\left\{\begin{gathered}\Dbin\!\left(\frac{1 + s}{1 + \sqrt{\frac{d - 1}{2d}}}\,\middle\|\,\frac{1}{1 + \sqrt{\frac{d - 1}{2d}}}\right),\\
\Dbin\!\left(\frac{d + 1 + 2N}{2d}\,\middle\|\,\frac{d + 1}{2d}\right)\end{gathered}\right\} .
\label{npt}
\ee
If $N > 0$ and $k \coloneqq \min\{\sqrt r,(d + 1)/2\}$, this further implies
\bb
\widetilde{L}(\rho) \geq 1 - h_2\!\left(\frac12 + \frac{N}{2k}\right) \geq 1 - h_2\!\left(\frac12 + \frac{N}{d + 1}\right)\, .
\label{npt_rank}
\ee
Both terms in~\eqref{npt} vanish when $N = 0$. If $d = 1$, the state is separable and $\widetilde{L}(\rho) = 0$.
\end{enumerate}
\end{prop}
\begin{proof}
For (i), a fixed-POVM $\widetilde{\ell}$ is a supremum of affine functions of $\rho$, while $\ell$ is convex because $p\log p$ is convex on $[0,1]$, with $0\log0 = 0$. Taking suprema over POVMs preserves convexity. The one-outcome POVM gives zero, and $\widetilde{\ell} \geq \ell$ follows from the dual representation in Lemma~\ref{povm_dual_lemma}.

For the upper bound, the trace inequality $\Tr[XY] \leq \Tr[|X||Y|]$ for Hermitian $X,Y$ follows by writing $X = X_+ - X_-$ and using $\pm Y \leq |Y|$. Consequently, for any positive weights, Jensen's inequality gives
\bb
\sum_j\Tr[\rho E_j]\log y_j & \leq \log\left(\sumno_j y_j\Tr[\rho E_j]\right)\\
& \leq \log\Tr\!\left[\big|\rho^{\pt}\big|\sumno_j y_j\big|E_j^{\pt}\big|\right]\\
& \leq \log\big\|\rho^{\pt}\big\|_1 + \log\left\|\sumno_j y_j\big|E_j^{\pt}\big|\right\|_\infty\, .
\label{fg_upper_bound}
\ee
Taking the two suprema proves~\eqref{fg_bounds}.

For (ii), let $\MM$ and $\pazocal{N}$ be represented by the POVMs $(E_i)_i$ and $(K_j)_j$ on $AB$ and $A'B'$, respectively. Their product POVM, with weights $y_i z_j$, has a logarithmic expectation equal to the sum of the two marginal logarithmic expectations. Its weighted absolute-operator sum is the tensor product of the two marginal sums, because
\bb
\big|(E_i\otimes K_j)^{\pt}\big| = \big|E_i^{\pt}\big|\otimes\big|K_j^{\pt}\big|\, .
\label{f_product_absolute_value}
\ee
Multiplicativity of the operator norm and optimisation prove the assertion for $\widetilde{L}$. For $L$, put $p_{ij} \coloneqq \Tr[\rho_{AA':BB'}(E_i\otimes K_j)]$, with marginals $p_i$ and $p'_j$. Direct expansion gives
\bb
\ell\big(\rho_{AA':BB'},\MM\otimes\pazocal{N}\big) = \ell(\rho_{A:B},\MM) + \ell(\rho_{A':B'},\pazocal{N}) + \sum_{i,j:\,p_{ij}>0}p_{ij}\log\frac{p_{ij}}{p_i p'_j}\, .
\label{g_product_expansion}
\ee
The last term is nonnegative, which proves the second assertion after optimisation. In particular, both functions are superadditive on tensor products and cannot increase when subsystems are discarded.

For (iii), the map $T \coloneqq \Gamma_{\rm out}\circ\Lambda\circ\Gamma_{\rm in}$ is a channel, so $T^\dag$ is unital and completely positive. For an output POVM $(E_j)_j$, its pullback $(\Lambda^\dag(E_j))_j$ has the same probabilities and satisfies
\bb
\big\|(\Lambda^\dag(E_j))^{\Gamma_{\rm in}}\big\|_\infty = \big\|T^\dag(E_j^{\Gamma_{\rm out}})\big\|_\infty \leq \big\|E_j^{\Gamma_{\rm out}}\big\|_\infty\, .
\label{g_pullback_norm}
\ee
Thus the pullback can only increase $\ell$, proving~\eqref{g_data_processing}. Zero pullback effects have zero probability and are omitted.

For (iv), both partial transposes are positive because $b<1$. If $a<1$, choose $y_1 = a/b$ and $y_0 = (1 - a)/(1 - b)$, so $y_1 > y_0 > 0$ and
\bb
\big\|y_1|A^{\pt}| + y_0|(\id-A)^{\pt}|\big\|_\infty & = \big\|y_0\id + (y_1 - y_0)A^{\pt}\big\|_\infty\\
& \leq by_1 + (1 - b)y_0 = 1\, .
\label{f_binary_norm}
\ee
Substitution in the definition of $\widetilde{\ell}$ gives~\eqref{f_binary}. If $a = 1$, take $y_1 = 1/b$ and $y_0\to 0^+$; the logarithmic expectation is $\log(1/b)$ and the norm has limit at most one, proving the same bound.

For (v), we may assume $N > 0$. Let $W \coloneqq (\rho^{\pt})_-$ and let $P$ be its support projector, so that $\Tr\big[\rho W^{\pt}\big] = -s^2$ and $\Tr\big[\rho P^{\pt}\big] = -N$. For a pure state $\psi = \ketbra{\psi}$ with Schmidt probabilities $(\lambda_i)_i$, the eigenvalues of $\psi^{\pt}$ are $\lambda_i$ and $\pm\sqrt{\lambda_i\lambda_j}$ for $i<j$. Consequently, putting $q \coloneqq \sqrt{(d - 1)/(2d)}$, we have
\bb
\big\|(\psi^{\pt})_-\big\|_2^2 = \frac{1-\sum_i\lambda_i^2}{2} \leq q^2,\qquad \big\|(\psi^{\pt})_+\big\|_2 \leq \big\|\psi^{\pt}\big\|_2 = 1\, .
\label{npt_pure_state_bounds}
\ee
Since $W \geq 0$, the Hilbert--Schmidt Cauchy--Schwarz inequality gives $-qs \leq \Tr\big[W\psi^{\pt}\big] \leq s$ for every pure $\psi$, hence $-qs\id \leq W^{\pt} \leq s\id$. Together with $0 \leq W \leq s\id$, this shows that
\bb
A &\coloneqq \frac{\id-W^{\pt}/s}{1 + q},\\
0 &\leq A \leq \id,\qquad 0 \leq A^{\pt} = \frac{\id-W/s}{1 + q} \leq \frac{\id}{1 + q},\\
\Tr[\rho A] &= \frac{1 + s}{1 + q}\, .
\label{npt_weighted_effect}
\ee
Taking the expectation of $W^{\pt} \geq -qs\id$ in $\rho$ yields $s \leq q$. Thus $0 < 1/(1 + q) < (1 + s)/(1 + q) \leq 1$, and item~(iv) gives
\bb
\widetilde{L}(\rho) \geq \Dbin\!\left(\frac{1 + s}{1 + q}\,\middle\|\,\frac{1}{1 + q}\right) .
\label{npt_weighted_threshold}
\ee
including the endpoint $s = q$.

Likewise, $\Tr\big[(\psi^{\pt})_\pm\big] = (\|\psi^{\pt}\|_1\pm 1)/2$ and $\|\psi^{\pt}\|_1 = (\sum_i\sqrt{\lambda_i})^2 \leq d$ imply
\bb
-\frac{d - 1}{2}\id \leq P^{\pt} \leq \frac{d + 1}{2}\id\, .
\label{npt_projector_bounds}
\ee
In particular, $N \leq (d - 1)/2$. The effect
\bb
B &\coloneqq \frac{(d + 1)\id/2-P^{\pt}}{d},\\
0 &\leq B \leq \id,\qquad 0 \leq B^{\pt} = \frac{(d + 1)\id/2-P}{d} \leq \frac{d + 1}{2d}\id,\\
\Tr[\rho B] &= \frac{d + 1+2N}{2d}
\label{npt_projector_effect}
\ee
therefore satisfies item~(iv), which yields
\bb
\widetilde{L}(\rho) \geq \Dbin\!\left(\frac{d + 1+2N}{2d}\,\middle\|\,\frac{d + 1}{2d}\right) .
\label{npt_negativity_threshold}
\ee
including the endpoint $N = (d - 1)/2$.

For the comparison, the symmetric effects $(\id - W^{\pt}/s)/2$ and $(\id - 2P^{\pt}/(d + 1))/2$ are obtained from $A$ and $B$, respectively, by retaining an accepting outcome with probability $(1 + q)/2$ or $d/(d + 1)$. These stochastic acceptance maps act on the binary state probabilities, so data processing gives
\bb
\Dbin\!\left(\frac{1 + s}{1 + q}\,\middle\|\,\frac{1}{1 + q}\right) & \geq \Dbin\!\left(\frac{1 + s}{2}\,\middle\|\,\frac12\right) = 1 - h_2\!\left(\frac{1 + s}{2}\right),\\
\Dbin\!\left(\frac{d + 1+2N}{2d}\,\middle\|\,\frac{d + 1}{2d}\right) & \geq \Dbin\!\left(\frac12+\frac{N}{d + 1}\,\middle\|\,\frac12\right) = 1 - h_2\!\left(\frac12+\frac{N}{d + 1}\right) .
\label{f_spectral_thinning}
\ee
Since $s \geq N/\sqrt r$, this implies~\eqref{npt_rank}.\footnote{The rank estimate can also be obtained directly: partial transpose preserves the Hilbert--Schmidt norm, so $\big\|P^{\pt}\big\|_\infty \leq \big\|P^{\pt}\big\|_2 = \sqrt r$, while the preceding dimension estimates give $\big\|P^{\pt}\big\|_\infty \leq (d + 1)/2$. Thus $\big\|P^{\pt}\big\|_\infty \leq k$. Since $k \geq 1$ and $N \leq k$, the effect $(\id-P^{\pt}/k)/2$ has partial transpose between zero and $\id/2$, and acceptance $(1 + N/k)/2$; item~(iv) gives the first expression in~\eqref{npt_rank}.}
\end{proof}

Single-copy monotonicity of $\widetilde{L}$ under arbitrary PPT channels is not asserted here. In particular, the pullback proof for $L$ does not apply to $\widetilde{L}$, since absolute value does not obey $|T^\dag(X)| \leq T^\dag(|X|)$ for general unital completely positive maps. The regularisation of $\widetilde{L}$ is a PPT monotone by its operational identification with $\Eppt$ below.

For $H = \widetilde{L},L$, item~(ii) makes $H(\rho^{\otimes n})$ superadditive in $n$, while item~(i) gives $H(\rho^{\otimes n}) \leq n\log\big\|\rho^{\pt}\big\|_1$. Fekete's lemma therefore proves that the limits in~\eqref{FG_regularisations} exist and equal the corresponding suprema.

For $N > 0$, let $P$ again denote the strictly negative spectral projector of $\rho^{\pt}$. The exact spectral extrema of $P^{\pt}$ can give a further improvement, although they also depend on the negative eigenspace of $\rho^{\pt}$. Put $c \coloneqq \max\{1,\lambda_{\max}(P^{\pt})\}$ and $z \coloneqq -\lambda_{\min}(P^{\pt})>0$. Since $0 < N \leq z$, the effect $(c\id - P^{\pt})/(c + z)$ gives
\bb
\widetilde{L}(\rho) \geq \Dbin\!\left(\frac{c + N}{c + z}\,\middle\|\,\frac{c}{c + z}\right) .
\label{npt_exact_projector}
\ee
Indeed its spectrum lies in $[0,1]$ and its partial transpose in $[0,c/(c + z)]$. Since $1 \leq c \leq k$, $0<z \leq k$, $c \leq (d + 1)/2$ and $z \leq (d - 1)/2$, the maps
\bb
t \longmapsto \frac{(c + z)t+k - c}{2k},\qquad t \longmapsto \frac{(c + z)t+(d + 1)/2-c}{d}
\label{npt_exact_comparison}
\ee
are binary postprocessings. Applied to the binary pair in~\eqref{npt_exact_projector}, the first produces the pair in~\eqref{npt_rank}, while the second produces the negativity-dependent pair in~\eqref{npt}. Data processing for binary relative entropy therefore proves both comparisons. No comparison with the first term in~\eqref{npt} or optimisation over all effects is asserted.

\subsection{Operational interpretation}

We now identify the regularised functions with the operational distillation rate. The achievability argument applies a fixed measurement repeatedly and selects its typical weighted outcomes; the reverse inequality follows by pulling back the Bell test of a distillation protocol.

\begin{manualthm}{\ref{regularised_formula_thm}}[(Regularised formula for the PPT distillable entanglement)]
For every finite-dimensional bipartite state, with $L$ defined by~\eqref{L_definition},
\begin{equation}
\begin{aligned}
\Eppt(\rho) = L^\infty(\rho) = \sup_{n\in\N^+}\frac1n\,L\big(\rho^{\otimes n}\big)\, .
\end{aligned}
\tag{\ref{regularisation}}
\end{equation}
\end{manualthm}
\begin{proof}
Fix a block length $m$ and a measurement channel $\MM$ on $m$ copies, represented by the POVM $(E_x)_x$ and outputting a classical label $x$. Fix weights $(y_x)_x>0$ and put $h\coloneqq\sum_x\Tr\!\left[\rho^{\otimes m}E_x\right]\log y_x$. On $n$ independent blocks, accept the words $x^n=(x_1,\ldots,x_n)$ for which $\sum_{i=1}^n\log y_{x_i}\geq n(h-\delta)$, with any fixed $\delta>0$. By the law of large numbers, their summed effect $Q$ has acceptance tending to one as $n\to\infty$. Wordwise spectral calculus gives
\bb
-\bigotimes_{i=1}^n\big|E_{x_i}^{\pt}\big| \leq \bigotimes_{i=1}^n E_{x_i}^{\pt} \leq \bigotimes_{i=1}^n\big|E_{x_i}^{\pt}\big|\, .
\label{wordwise_spectral_bound}
\ee
Each accepted word has weight at least $2^{n(h - \delta)}$. Summing the corresponding positive bounds, and then adding the omitted positive words, yields
\bb
\big\|Q^{\pt}\big\|_\infty \leq 2^{-n(h - \delta)}\left\|\sumno_x y_x\big|E_x^{\pt}\big|\right\|_\infty^n\, .
\label{weighted_typical_bound}
\ee
Letting $\delta\to0^+$, equation~\eqref{channel} therefore achieves every positive rate below the objective in~\eqref{povm_f}, evaluated on $\rho^{\otimes m}$, divided by $m$. Discarding fewer than $m$ unused copies covers arbitrary total lengths. Optimising the fixed weights and measurement gives $\Eppt(\rho) \geq \widetilde{L}(\rho^{\otimes m})/m$, and hence $\Eppt(\rho) \geq \widetilde{L}^{\,\infty}(\rho) \geq L^\infty(\rho)$.

Conversely, pull back the Bell projection of an achieving protocol to $Q_n$. If its overlap is $1-\e_n$ and output rank is $d_n$, the measurement represented by $(Q_n,\id - Q_n)$ gives
\bb
L\big(\rho^{\otimes n}\big) \geq (1-\e_n)\log d_n - h_2(\e_n) - \e_n\log(1 + 1/d_n)\, .
\label{L_protocol_bound}
\ee
Indeed~\eqref{effect} and the triangle inequality bound the two partial-transpose norms by $1/d_n$ and $1 + 1/d_n$. Divide by $n$, let $\e_n\to0$, and take the supremum over achievable rates to obtain $L^\infty(\rho) \geq \Eppt(\rho)$. The superadditivity in Proposition~\ref{FG_properties_prop}(ii) and the linear bound in item~(i) give the limits and supremum formulae by Fekete's lemma.
\end{proof}

The same proof establishes the useful auxiliary identity
\bb
\widetilde{L}^{\,\infty}(\rho) = \Eppt(\rho)\, ,\qquad \widetilde{L}(\rho) \leq \Eppt(\rho)\, .
\label{auxiliary_regularisation}
\ee

Together with Proposition~\ref{FG_properties_prop}(v), Theorem~\ref{regularised_formula_thm} turns~\eqref{npt} into an explicit lower bound on $\Eppt(\rho)$ in terms of the spectrum of $\rho^{\pt}$, with the negativity relaxation in~\eqref{npt_rank}. In particular, every NPT state has a strictly positive rate. This makes quantitative the qualitative distillability theorem of Eggeling, Vollbrecht, Werner and Wolf~\cite[Theorem~1]{Eggeling2001}. The regularised functions are both monotone under PPT channels, either by their operational interpretation or by data processing for $L$ on every tensor power.

\section{Converse statements}

The regularised characterisation concerns vanishing-error rates. By retaining finite-block success bounds in the converse arguments, we also obtain exponential strong converses. We present the two converse constructions in separate subsections.

\subsection{Converse via quadratic forms on operator spaces}

The main result of this subsection is the conditional bound in Theorem~\ref{quadratic_converse_thm}. We first establish the data-processing inequality on which its proof is based.

\subsubsection{Data processing and the conditional converse}

For a positive operator $\omega \geq 0$, recall that $\LL_\omega(X) = \omega X$ denotes left multiplication on the Hilbert--Schmidt operator space. Note that, with respect to the Hilbert--Schmidt scalar product, we have $\LL_\omega \geq 0$. In what follows, positivity of a super-operator will always refer to the Hilbert--Schmidt scalar product. For two positive definite (super-)operators $A,B$, write
\bb
A\#_t B \coloneqq A^{1/2}\big(A^{-1/2}BA^{-1/2}\big)^tA^{1/2}
\label{weighted_geometric_mean}
\ee
for their $t$-weighted geometric mean~\cite{BHATIA}. For $t\in[0,1]$, we can extend the definition to all positive semi-definite $A,B \geq 0$ by setting $A\#_tB \coloneqq \lim_{\e\to 0^+} (A + \e\id)\#_t (B + \e\id)$. Outside that range we use the expression only when the displayed power is defined; in particular, $\KK\#_{1 + s}\LL_\omega$ is well-defined for $\KK>0$, $\omega \geq 0$ and $s>0$.

For $V_\alpha$ defined in~\eqref{V_definition}, polar decomposition gives the equivalent representation
\bb
V_\alpha[\omega,\KK] = \frac{1}{\alpha\ln 2}\lim_{t\to 0^+}\frac1t\left(\Tr[\omega] - \Tr\big[(\LL_\omega\#_t\KK)(\id)\big]\right) .
\label{V_geometric_derivative}
\ee
The logarithm is restricted to the support of the super-operator on which it acts. This causes no difficulty, since its kernel coincides with that of $\LL_{\sqrt\omega}$, and $\sqrt\omega$ belongs to the orthogonal complement of that kernel. This representation remains valid for singular $\omega$; the details, including continuity, are given in Appendix~\ref{quadratic_support_app}.

This expression makes it especially simple to prove data processing. We recall the properties that we need.

\begin{lemma}[(Transformer inequality~{\cite[Chapter~4]{BHATIA}})]\label{geometric_transformer_lemma}
For all positive semi-definite operators $A,B$, all linear maps $M$ between finite-dimensional Hilbert spaces, and all $t\in[0,1]$, it holds that
\bb
\big(MA M^\dag\big)\#_t\big(MB M^\dag\big) \geq M(A\#_tB)M^\dag\, ,
\label{transformer}
\ee
with equality if $M$ is invertible. Moreover, the weighted geometric mean is jointly monotone and jointly concave in its two arguments. Singular arguments are understood by decreasing positive regularisation.
\end{lemma}

\begin{lemma}[(Data processing)]\label{quadratic_data_processing_lemma}
Let $\Lambda$ be a quantum channel, and let $\KK,\KK'>0$ be super-operators on its input and output operator spaces, respectively, such that
\bb
\Lambda\circ\KK\circ\Lambda^\dag \leq \KK'\, .
\label{paired_form_condition}
\ee
Then, for every $\omega \geq 0$ and every $\alpha>0$, we have
\bb
V_\alpha[\omega,\KK] \geq V_\alpha[\Lambda(\omega),\KK']\, .
\label{paired_data_processing}
\ee
The same statement holds for channels between finite direct sums of matrix algebras.
\end{lemma}
\begin{proof}
The Schwarz inequality for the unital completely positive map $\Lambda^\dag$ gives
\bb
\Braket{Z,\LL_{\Lambda(\omega)}[Z]} & = \Tr Z^\dag\Lambda(\omega)Z\\
& = \Tr\omega\Lambda^\dag(ZZ^\dag)\\
& \geq \Tr\omega\Lambda^\dag(Z)\Lambda^\dag(Z)^\dag\\
& = \Braket{Z,\big(\Lambda\circ\LL_\omega\circ\Lambda^\dag\big)[Z]} .
\label{Schwarz_left_multiplication}
\ee
Hence, for all $t\in[0,1]$, Lemma~\ref{geometric_transformer_lemma} implies that
\bb
\LL_{\Lambda(\omega)}\#_t\KK' &\geqt{(i)} \big(\Lambda\circ\LL_\omega\circ\Lambda^\dag\big)\#_t \big(\Lambda\circ\KK\circ\Lambda^\dag\big)\\
&\geqt{(ii)} \Lambda\circ(\LL_\omega\#_t\KK)\circ\Lambda^\dag\, ,
\label{geometric_data_processing}
\ee
where~(i) uses joint monotonicity, and~(ii) is the transformer inequality. Taking the expectation in the identity operator, and remembering that $\Lambda^\dag(\id) = \id$, we have
\bb
\Tr\big[(\LL_{\Lambda(\omega)}\#_t\KK')[\id]\big] \geq \Tr\big[(\LL_\omega\#_t\KK)[\id]\big] .
\ee
At $t = 0$ both sides equal $\Tr\omega$. Subtracting from this common value, dividing by $\alpha t\ln 2$, and taking $t\to 0^+$ proves~\eqref{paired_data_processing}.
\end{proof}

We will also use three elementary properties of~\eqref{V_definition}. Joint concavity of the geometric mean implies joint convexity of $V_\alpha$ in $(\omega,\KK)$. For $\lambda,\mu>0$, homogeneity gives
\bb
V_\alpha[\lambda\omega,\mu\KK] = \lambda V_\alpha[\omega,\KK] + \frac{\lambda\Tr\omega}{\alpha}\log\frac{\lambda}{\mu}\, .
\label{V_scaling}
\ee
Finally, the tensor-product identity for geometric means gives, for normalised states $\omega_1,\omega_2$,
\bb
V_\alpha[\omega_1\otimes\omega_2,\KK_1\otimes\KK_2] = V_\alpha[\omega_1,\KK_1] + V_\alpha[\omega_2,\KK_2]\, .
\label{V_additivity}
\ee
Direct sums satisfy the analogous identity, without normalisation of the summands. Here and below a direct-sum output is viewed as a classical register together with its conditional quantum systems; only its block-diagonal operator algebra is needed.

\begin{lemma}\label{conditional_test_lemma}
Let $\omega = \omega_{AA':BB'}$ be a state, and let $\KK>0$ and $\JJ>0$ act on the operator spaces of $AA':BB'$ and $A':B'$, respectively. Assume that, for some $\alpha>0$,
\bb
\Braket{Q\otimes X,\KK[Q\otimes X]} \leq \big\|Q^{\pt}\big\|_\infty^\alpha\Braket{X,\JJ[X]}
\label{ancillary_form_one_copy}
\ee
for every test $0 \leq Q \leq \id$ on $AB$ and every operator $X$ on $A'B'$. Then, for every nonzero such test, putting $p \coloneqq \Tr\omega(Q\otimes\id)$, we have
\bb
V_\alpha[\omega_{AA':BB'},\KK] - V_\alpha[\omega_{A':B'},\JJ] \geq p\log\frac1{\big\|Q^{\pt}\big\|_\infty} - (1 - p)\log\left(1 + \big\|Q^{\pt}\big\|_\infty\right) - \frac2\alpha h_2(p)\, .
\label{conditional_test}
\ee
\end{lemma}
\begin{proof}
First assume that $0<p<1$. Consider the channel that measures $\{Q,\id - Q\}$ on $AB$ and retains the outcome together with $A'B'$. Its adjoint and its action on $\omega$ are
\bb
\Lambda^\dag(X_0\oplus X_1)& = Q\otimes X_0 + (\id - Q)\otimes X_1\, ,\\
\Lambda(\omega)& = p\,\omega^{(0)}_{A':B'}\oplus(1 - p)\,\omega^{(1)}_{A':B'}\, ,
\label{conditional_measurement}
\ee
where $\omega^{(x)}_{A':B'}$ is the normalised reduced state conditioned on the outcome $x\in\{0,1\}$, where $x = 0$ corresponds to $Q$ and $x = 1$ to $\id - Q$. In particular, $p\,\omega^{(0)}_{A':B'} + (1 - p)\,\omega^{(1)}_{A':B'} = \omega_{A':B'}$. The squared norm associated with $\KK$ satisfies
\bb
&\Braket{X_0\oplus X_1,\big(\Lambda\circ\KK\circ\Lambda^\dag\big)[X_0\oplus X_1]}\\
&\qquad = \Braket{Q\otimes X_0 + (\id - Q)\otimes X_1,\, \KK\big[Q\otimes X_0 + (\id - Q)\otimes X_1\big]}\\
&\qquad = \Braket{p\, \frac{Q\otimes X_0}{p} + (1 - p)\, \frac{(\id - Q)\otimes X_1}{1 - p},\, \KK\Big[p\, \frac{Q\otimes X_0}{p} + (1 - p)\, \frac{(\id - Q)\otimes X_1}{1 - p}\Big]}\\
&\qquad \leqt{(i)} p\Braket{\frac{Q\otimes X_0}{p},\KK\Big[\frac{Q\otimes X_0}{p}\Big]} + (1 - p) \Braket{\frac{(\id - Q)\otimes X_1}{1 - p},\KK\Big[\frac{(\id - Q)\otimes X_1}{1 - p}\Big]}\\
&\qquad = \frac1p\Braket{Q\otimes X_0,\KK[Q\otimes X_0]} + \frac1{1 - p}\Braket{(\id - Q)\otimes X_1,\KK[(\id - Q)\otimes X_1]}\\
&\qquad \leqt{(ii)} \frac{\big\|Q^{\pt}\big\|_\infty^\alpha}{p}\Braket{X_0,\JJ[X_0]} + \frac{\big(1 + \big\|Q^{\pt}\big\|_\infty\big)^\alpha}{1 - p} \Braket{X_1,\JJ[X_1]}\\
&\qquad \eqt{(iii)} \Braket{X_0\oplus X_1,\, \KK'[X_0\oplus X_1]} ,
\label{measurement_form_bound}
\ee
where the inequality in~(i) is convexity of a squared norm, applied with weights $p$ and $1 - p$, that in~(ii) follows from~\eqref{ancillary_form_one_copy} and the fact that $\big\|\id - Q^{\pt}\big\|_\infty \leq 1 + \big\|Q^{\pt}\big\|_\infty$, and in~(iii) we defined
\bb
\KK'[X_0\oplus X_1] \coloneqq \frac{\big\|Q^{\pt}\big\|_\infty^\alpha}{p} \JJ[X_0] \oplus \frac{\left(1 + \left\|Q^{\pt}\right\|_\infty\right)^\alpha}{1 - p} \JJ[X_1]\, .
\ee
Now,
\bb
V_\alpha[\omega_{AA':BB'},\KK] \ &\geqt{(iv)} \ V_\alpha\big[p\,\omega^{(0)}_{A':B'} \oplus (1 - p)\,\omega^{(1)}_{A':B'},\, \KK' \big]\\
&\eqt{(v)} \ V_\alpha\!\left[p\,\omega^{(0)}_{A':B'},\frac{\big\|Q^{\pt}\big\|_\infty^\alpha}{p}\JJ\right] + V_\alpha\!\left[(1 - p)\,\omega^{(1)}_{A':B'},\frac{\big(1 + \big\|Q^{\pt}\big\|_\infty\big)^\alpha}{1 - p}\JJ\right]\\
&\eqt{(vi)} \ p\,V_\alpha\big[\omega^{(0)}_{A':B'},\JJ\big] + (1 - p)\,V_\alpha\big[\omega^{(1)}_{A':B'},\JJ\big]\\
&\quad\ + \frac1\alpha\left(p\log\frac{p^2}{\big\|Q^{\pt}\big\|_\infty^\alpha} + (1 - p)\log\frac{(1 - p)^2}{\big(1 + \big\|Q^{\pt}\big\|_\infty\big)^\alpha}\right)\\
&\geqt{(vii)} \ V_\alpha[\omega_{A':B'},\JJ] + p\log\frac1{\big\|Q^{\pt}\big\|_\infty} - (1 - p)\log\big(1 + \big\|Q^{\pt}\big\|_\infty\big) - \frac2\alpha h_2(p)\, .
\ee
Here, in~(iv) we used Lemma~\ref{quadratic_data_processing_lemma}, (v)~follows from the definition of~$\KK'$, (vi)~descends from~\eqref{V_scaling}, and~(vii) uses convexity in the physical state. The cases $p = 0,1$ follow by replacing $Q$ with $(1 - \e)Q + \e\id/2$ and taking $\e\to0^+$.
\end{proof}

For $Q = \Phi_d$ on $AB$, the two states $\omega^{(0)}_{A':B'}$ and $\omega^{(1)}_{A':B'}$ are precisely those obtained by projecting onto the maximally entangled state or its complement. Since $\big\|\Phi_d^{\pt}\big\|_\infty = 1/d$, the preceding lemma gives
\bb
V_\alpha[\omega_{AA':BB'},\KK] - V_\alpha[\omega_{A':B'},\JJ] \geq p\log d - (1 - p)\log(1 + 1/d) - \frac2\alpha h_2(p)\, .
\label{bell_chain}
\ee
In particular, $\alpha = 2$ gives the finite-error Bell inequality for the state-weighted forms considered below.

\begin{manualthm}{\ref{quadratic_converse_thm}}[(Conditional quadratic converse)]
Let $\rho = \rho_{A:B}$ be a bipartite state. With $V_\alpha$ defined by~\eqref{V_definition}, put
\begin{equation}
\begin{aligned}
E_Q(\rho) = E_Q(\rho_{A:B}) \coloneqq \inf_{\rho_{AA':BB'},\,\alpha,\KK,\JJ} \left\{V_\alpha[\rho_{AA':BB'},\KK] - V_\alpha[\rho_{A':B'},\JJ]\right\} ,
\end{aligned}
\tag{\ref{quadratic_bound}}
\end{equation}
where the infimum is over all finite-dimensional extensions $\rho_{AA':BB'}$ of $\rho_{A:B}$, all $\alpha>0$, and all strictly positive super-operators $\KK,\JJ$ on the operator spaces of $AA'BB'$ and $A'B'$, respectively, satisfying
\begin{equation}
\begin{aligned}
\Braket{Q_n\otimes X_n,\KK^{\otimes n}[Q_n\otimes X_n]} \leq \big\|Q_n^{\Gamma_{B^n}}\big\|_\infty^\alpha\Braket{X_n,\JJ^{\otimes n}[X_n]}
\end{aligned}
\tag{\ref{ancillary_form_all_copies}}
\end{equation}
for all $n\in\N^+$, all tests $0 \leq Q_n \leq \id$ on $A^nB^n$, and all operators $X_n$ on $(A'B')^n$. Then
\begin{equation}
\begin{aligned}
\Eppt(\rho_{A:B}) \leq \Eppt^\dag(\rho_{A:B}) \leq E_Q(\rho_{A:B})\, .
\end{aligned}
\tag{\ref{conditional}}
\end{equation}
Every rate strictly above $E_Q(\rho_{A:B})$ has an exponential strong converse. 
\end{manualthm}
\begin{proof}
We first prove the ordinary converse. Fix an extension $\rho_{AA':BB'}$ and $\alpha,\KK,\JJ$ as in~\eqref{quadratic_bound}. Let $Q_n$ be the test associated with a distillation protocol of output rank $d_n$, and write $\Tr\rho_{AB}^{\otimes n}Q_n = 1 - \e_n$. By~\eqref{effect}, $\big\|Q_n^{\Gamma_{B^n}}\big\|_\infty \leq d_n^{-1}$. Applying Lemma~\ref{conditional_test_lemma} to $\rho_{AA':BB'}^{\otimes n}$, and using~\eqref{V_additivity}, gives
\bb
&n\big(V_\alpha[\rho_{AA':BB'},\KK] - V_\alpha[\rho_{A':B'},\JJ]\big)\\
&\qquad \geq (1 - \e_n)\log d_n - \e_n\log(1 + 1/d_n) - \frac2\alpha h_2(\e_n)\, .
\label{conditional_finite_error}
\ee
Dividing by $n$ and taking $n\to\infty$ bounds $\Eppt(\rho_{A:B})$ by the fixed conditional difference whenever $\e_n\to 0$. Taking the infimum in~\eqref{quadratic_bound} proves the ordinary converse. The exponential strong converse for each fixed difference follows by retaining a small positive power and spectrally truncating on the ancillary operator space, as shown in Appendix~\ref{conditional_strong_converse_app}. For every rate strictly above $E_Q(\rho_{A:B})$, one can choose a fixed difference below that rate, which proves the remaining assertion.
\end{proof}

The hypothesis~\eqref{ancillary_form_all_copies} concerns arbitrary ancillary operators, not just ancillary effects. Neither this hypothesis nor its tensor-power version follows from a one-copy bound on the physical tests alone.

For example, take a qubit with trivial second party and no ancilla, and let
\bb
\KK(X) \coloneqq \frac32 X - \frac12\Tr[X]\id\, .
\label{non_tensorising_form}
\ee
This super-operator has eigenvalue $1/2$ on the identity and $3/2$ on traceless operators, hence is strictly positive. Denoting the eigenvalues of a qubit effect by $0\leq b\leq a\leq1$, we have
\bb
\Braket{Q,\KK(Q)} = a^2 + b^2 - ab \leq a^2 = \|Q\|_\infty^2\, .
\label{one_copy_form_example}
\ee
Nevertheless, the Bell projection on two copies satisfies
\bb
\Braket{\Phi_2,\KK^{\otimes2}(\Phi_2)} = \frac14\left[\left(\frac12\right)^2 + 3\left(\frac32\right)^2\right] = \frac74 > 1 = \|\Phi_2\|_\infty^2\, .
\label{two_copy_form_failure}
\ee
Thus even strict positivity of the form and validity for every one-copy effect do not justify tensorisation.

\subsubsection{Retrieving and extending the Rains bound}

With trivial ancillary systems and $\JJ = \Id_{\C}$, Theorem~\ref{quadratic_converse_thm} gives $\Eppt^\dag(\rho) \leq V_\alpha[\rho,\KK]$ whenever $\Braket{Q_n,\KK^{\otimes n}[Q_n]} \leq \big\|Q_n^{\pt}\big\|_\infty^\alpha$ for all $n\in\N^+$ and all tests $0 \leq Q_n \leq \id$. Indeed, the ancillary contribution is $V_\alpha[1,\Id_{\C}] = 0$. We record two consequences.

It is interesting to observe that the (strong converse!) Rains bound can be obtained as a special case of Theorem~\ref{quadratic_converse_thm} with trivial ancillary systems. Namely, consider a strictly positive operator $\sigma>0$ in the Rains set, and define $\KK[X] = \RR_\sigma[X] \coloneqq X\sigma$ as the right multiplication by $\sigma$. Then, for all $n\in \N^+$ and all tests $0 \leq Q_n \leq \id$, it holds that
\bb
\Braket{Q_n,\KK^{\otimes n}[Q_n]} = \Tr \sigma^{\otimes n} Q_n^2 \leq \Tr \sigma^{\otimes n} Q_n = \Tr \big(\sigma^{\pt}\big)^{\otimes n} Q_n^{\pt} \leq \big\|Q_n^{\pt}\big\|_\infty\, .
\label{rains_quadratic_estimate}
\ee
Hence, Theorem~\ref{quadratic_converse_thm} and the commutativity between left and right multiplication gives
\bb
\Eppt(\rho) \leq \Eppt^\dag(\rho) & \leq \Braket{\sqrt\rho, \big( \log \big( \LL_{\sqrt\rho} \circ \RR_\sigma^{-1} \circ \LL_{\sqrt\rho}\big)\big)\big[\sqrt\rho\big]}\\
& = \Braket{\sqrt\rho, \big( \log \LL_{\rho} - \log \RR_\sigma \big)\big[\sqrt\rho\big]}\\
& = D(\rho\|\sigma)\, .
\label{rains_relative_entropy}
\ee

The same conclusion holds for a singular Rains operator by taking a strictly positive approximation within the Rains set. There is also an ancillary extension of this construction. Define
\bb
R_{\mathrm{ext}}(\rho) = R_{\mathrm{ext}}(\rho_{A:B}) \coloneqq \inf_{\substack{\rho_{ABA'B'}:\,\Tr_{A'B'}\rho_{ABA'B'} = \rho_{AB}\\
\sigma_{ABA'B'}>0:\,\Tr\sigma_{ABA'B'} = 1}} \bigg\{&D(\rho_{ABA'B'}\|\sigma_{ABA'B'})\\
& - \rel{D}{\rho_{A'B'}}{\Tr_{AB}\big|\sigma_{ABA'B'}^{\pt_B}\big|}\bigg\} .
\label{extended_rains_definition}
\ee
The ancillary systems are arbitrary and finite, and $\pt_B$ transposes only the principal system $B$, not $B'$. The second reference is strictly positive since $\Tr_{AB}\big|\sigma^{\pt_B}\big| \geq \sigma_{A'B'} > 0$. Normalising $\sigma$ does not restrict the infimum: a common rescaling of both references cancels in the difference.

\begin{prop}\label{extended_rains_prop}
For every bipartite state $\rho$,
\bb
R^\infty(\rho) \leq R_{\mathrm{ext}}(\rho) \leq R(\rho), \qquad E_Q(\rho) \leq R_{\mathrm{ext}}(\rho)\, .
\label{extended_rains_sandwich}
\ee
Consequently, $R_{\mathrm{ext}}(\rho) = R(\rho)$ whenever $R(\rho) = R^\infty(\rho)$, and $\lim_{n\to\infty}\frac1n R_{\mathrm{ext}}(\rho^{\otimes n}) = R^\infty(\rho)$.
\end{prop}
\begin{proof}
Fix an extension and reference in~\eqref{extended_rains_definition}, and write $\tau_{A'B'} = \Tr_{AB}\big|\sigma_{ABA'B'}^{\pt_B}\big|$. Set $\KK = \RR_\sigma$, $\JJ = \RR_\tau$, and $\alpha = 1$. For every test $0 \leq Q \leq \id$ on $AB$ and every operator $X$ on $A'B'$, we have
\bb
\Braket{Q\otimes X,\KK[Q\otimes X]} & = \Tr\sigma(Q^2\otimes X^\dag X)\\
& \leq \Tr\sigma(Q\otimes X^\dag X)\\
& = \Tr\sigma^{\pt_B}(Q^{\pt_B}\otimes X^\dag X)\\
& \leq \big\|Q^{\pt_B}\big\|_\infty\Tr\big|\sigma^{\pt_B}\big|(\id\otimes X^\dag X)\\
& = \big\|Q^{\pt_B}\big\|_\infty\Braket{X,\JJ[X]} .
\label{extended_rains_compatibility}
\ee
The penultimate inequality follows by splitting $\sigma^{\pt_B}$ into its positive and negative parts. Since
\bb
\Tr_{A^nB^n}\left|(\sigma^{\otimes n})^{\Gamma_{B^n}}\right| = \big(\Tr_{AB}\big|\sigma^{\pt_B}\big|\big)^{\otimes n} = \tau^{\otimes n}\, ,
\label{extended_rains_tensor_reference}
\ee
the same estimate holds for arbitrary $Q_n,X_n$, proving~\eqref{ancillary_form_all_copies}. Equation~\eqref{rains_relative_entropy} identifies the two potentials with the relative entropies in~\eqref{extended_rains_definition}; hence $E_Q \leq R_{\mathrm{ext}}$ by~\eqref{quadratic_bound}. Trivial ancillas give
\bb
R_{\mathrm{ext}}(\rho) \leq \inf_{\sigma>0:\,\Tr\sigma = 1}\left\{D(\rho\|\sigma) + \log\big\|\sigma^{\pt}\big\|_1\right\} = R(\rho)\, ,
\ee
where the equality is the rescaling of the Rains set in~\eqref{rains_definition}.

For the opposite inequality, take $H>0$ on $AB$ with $\big\|H^{\pt_B}\big\|_\infty \leq 1$, and measure its spectral projections, retaining the outcome and $A'B'$. Writing $H = \sum_x h_xP_x$, the unnormalised ancillary blocks are
\bb
\rho^{(x)}_{A'B'} = \Tr_{AB}\big[(P_x\otimes\id)\rho_{ABA'B'}(P_x\otimes\id)\big],\qquad \sigma^{(x)}_{A'B'} = \Tr_{AB}\big[(P_x\otimes\id)\sigma_{ABA'B'}(P_x\otimes\id)\big] .
\ee
The same trace-norm argument as above gives $\sum_x h_x\sigma^{(x)}_{A'B'} \leq \tau_{A'B'}$. Relative-entropy data processing under the measurement and then under discarding its outcome, followed by operator monotonicity of the logarithm, yields
\bb
D(\rho_{ABA'B'}\|\sigma) & \geq \sum_x\rel{D}{\rho^{(x)}_{A'B'}}{\sigma^{(x)}_{A'B'}}\\
& = \sum_x\rel{D}{\rho^{(x)}_{A'B'}}{h_x\sigma^{(x)}_{A'B'}} + \Tr\rho_{AB}\log H\\
& \geq \rel{D}{\rho_{A'B'}}{\sumno_x h_x\sigma^{(x)}_{A'B'}} + \Tr\rho_{AB}\log H\\
& \geq D(\rho_{A'B'}\|\tau_{A'B'}) + \Tr\rho_{AB}\log H\, .
\ee
Taking the supremum over $H$ and invoking~\eqref{measured_rains_norm} bounds the fixed conditional difference below by $R^M(\rho_{AB})$. Now apply this inequality to the $n$-fold tensor powers of the same extension and reference. By~\eqref{extended_rains_tensor_reference} and additivity of relative entropy,
\bb
D(\rho_{ABA'B'}\|\sigma) - D(\rho_{A'B'}\|\tau) \geq \frac1n R^M(\rho_{AB}^{\otimes n})\, .
\ee
Letting $n\to\infty$, Lemma~\ref{measured_rains_lemma} gives $R^\infty(\rho_{AB})$. Only then do we infimise over the original extension and reference. The final assertions follow from the sandwich and weak additivity of $R^\infty$.
\end{proof}

Thus this ancillary relative-entropy construction cannot improve the regularised Rains converse. Whether the first inequality in~\eqref{extended_rains_sandwich} is always an equality before regularisation is not established here. At the opposite end, we have the following different family of quadratic converses.

\begin{cor}\label{two_state_forms_cor}
Let $\tau,\omega>0$ be two strictly positive states on a bipartite quantum system $AB$, and let $\theta\in [0,1]$. Then, for all states $\rho = \rho_{AB}$, we have
\bb
\Eppt(\rho) \leq \Eppt^\dag(\rho) \leq \frac12 \Braket{\sqrt\rho, \big( \log \big( \LL_{\sqrt\rho} \circ \Gamma \circ \LL_{\tau^{-\theta}}\circ \RR_{\omega^{-1 + \theta}} \circ \Gamma \circ \LL_{\sqrt\rho}\big)\big)\big[\sqrt\rho\big]} .
\ee
\end{cor}

\begin{proof}
Set
\bb
\KK \coloneqq \Gamma \circ \LL_{\tau^{\theta}}\circ \RR_{\omega^{1 - \theta}} \circ \Gamma\, .
\ee
Then, for all $n\in \N^+$ and all tests $0 \leq Q_n \leq \id$, we have
\bb
\Braket{Q_n, \KK^{\otimes n}[Q_n]} & = \Tr \left[ Q_n^{\pt} \big(\tau^{\otimes n}\big)^\theta Q_n^{\pt} \big(\omega^{\otimes n}\big)^{1 - \theta} \right]\\
& \leq \big\|Q_n^{\pt}\big\|_\infty^2\, \Big\|\big(\tau^{\otimes n}\big)^\theta\Big\|_{1/\theta} \Big\|\big(\omega^{\otimes n}\big)^{1 - \theta}\Big\|_{1/(1 - \theta)}\\
& = \big\|Q_n^{\pt}\big\|_\infty^2\, .
\ee
Hence, we can apply Theorem~\ref{quadratic_converse_thm} with trivial ancillary systems and $\alpha = 2$, which concludes the proof.
\end{proof}

\subsubsection{A different quadratic converse}

A particularly simple choice of quadratic form is $\KK = \Gamma\circ\LL_\tau\circ\Gamma$, where $\tau>0$ is a normalised state on $AB$. It gives the following single-copy quantity, written explicitly in terms of the auxiliary state:
\bb
U(\rho) \coloneqq \frac12\inf_{\tau>0}\Braket{\sqrt\rho,\,\left(\log\left(\LL_{\sqrt\rho}\circ\Gamma\circ\LL_{\tau^{-1}}\circ\Gamma\circ\LL_{\sqrt\rho}\right)\right)(\sqrt\rho)}\, .
\label{one_state_quadratic_bound}
\ee
The logarithm is restricted to the support of its argument, as in~\eqref{V_definition}, so singular physical states are allowed. In particular, $U(\rho) = \inf_{\tau>0}V_2[\rho,\Gamma\circ\LL_\tau\circ\Gamma]$.

For every $n\in\N^+$ and every effect $Q_n$, the corresponding form obeys
\bb
\Braket{Q_n,(\Gamma\circ\LL_\tau\circ\Gamma)^{\otimes n}(Q_n)} = \Tr\!\left[\tau^{\otimes n}\big(Q_n^{\pt}\big)^2\right] \leq \big\|Q_n^{\pt}\big\|_\infty^2\, .
\label{state_form_compatibility}
\ee
Theorem~\ref{quadratic_converse_thm}, with trivial ancillary systems and $\alpha = 2$, therefore gives
\bb
\Eppt(\rho) \leq \Eppt^\dag(\rho) \leq E_Q(\rho) \leq U(\rho) \leq \log\big\|\rho^{\pt}\big\|_1\, .
\label{U_bounds}
\ee
Every rate strictly above $U(\rho)$ has an exponential strong converse. To obtain the last inequality in~\eqref{U_bounds}, apply scalar Jensen to the logarithmic expectation in~\eqref{one_state_quadratic_bound}:
\bb
U(\rho) \leq \inf_{\tau>0}\frac12\log\Tr\!\left[\rho^{\pt}\tau^{-1}\rho^{\pt}\right] = \log\big\|\rho^{\pt}\big\|_1\, .
\label{log_negativity_comparison}
\ee
For the equality, Cauchy--Schwarz bounds the trace below by $\big\|\rho^{\pt}\big\|_1^2$, and the infimum is approached by taking $\tau$ proportional to $\big|\rho^{\pt}\big|$, with a strictly positive approximation when needed. In particular, $U$ is nonnegative and vanishes on PPT states.

Product auxiliary states and~\eqref{V_additivity} make $U$ subadditive. Applying the same converse to blocks thus gives
\bb
\Eppt^\dag(\rho) \leq \lim_{n\to\infty}\frac1nU(\rho^{\otimes n}) = \inf_{n\in\N^+}\frac1nU(\rho^{\otimes n}) \leq U(\rho)\, .
\label{U_regularisation}
\ee
The limit exists by subadditivity and nonnegativity. Every rate strictly above this infimum also has an exponential strong converse, by choosing one fixed block and one auxiliary state whose bound lies below the rate.

\subsubsection{Which quadratic forms can improve the bound?}

The definition of $E_Q$ in~\eqref{quadratic_bound} allows quadratic estimates that hold only for effects on the principal system. With trivial ancillary systems, this is weaker than an estimate on arbitrary operators. The distinction becomes important when comparing different quadratic forms. The relevant general result is the little noncommutative Grothendieck inequality.

\begin{thm}[{\cite[Theorem~8.1]{Pisier2012}}] \label{little_Grothendieck_thm}
Let $T:\mathcal{L}(\HH)\to\mathcal{L}(\HH')$ be a linear map, where the output operator space is equipped with the Hilbert--Schmidt norm $\|\cdot\|_2$. There exist states $\tau_0,\tau_1$ on $\HH$ such that, for every operator $X$,
\bb
\|T(X)\|_2^2 \leq \|T\|_{\infty\to2}^2 \left(\Tr\big[\tau_0 X^\dag X\big] + \Tr\big[\tau_1 XX^\dag\big]\right)\, , \qquad \|T\|_{\infty\to2} \coloneqq \sup_{\|X\|_\infty \leq 1}\|T(X)\|_2\, .
\label{little_Grothendieck}
\ee
\end{thm}

The states in this statement depend on $T$, but not on $X$. For Hermitian $X$, their average controls the quadratic form by a single state-weighted second moment, with a factor of two. For Werner states, this yields the following restriction on the possible improvements.

\begin{prop}\label{Grothendieck_Werner_prop}
Let $\rho$ be a full-rank Werner state, and let $\KK>0$ satisfy
\bb
\Braket{X,\KK^{\otimes n}[X]} \leq \big\|X^{\pt}\big\|_\infty^2 \qquad\forall\,n\in\N^+\, ,\quad\forall\,X\in\mathcal{L}(\HH_{AB}^{\otimes n})\, .
\label{full_operator_hypothesis}
\ee
Then, for every $n\in\N^+$,
\bb
U(\rho^{\otimes n}) \leq n V_2[\rho,\KK] + \frac12\, .
\label{Grothendieck_finite}
\ee
Consequently,
\bb
\lim_{n\to\infty}\frac1n U(\rho^{\otimes n}) \leq V_2[\rho,\KK]\, .
\label{Grothendieck_asymptotic}
\ee
\end{prop}

The proof is given in Appendix~\ref{Grothendieck_app}. The point is to apply the Grothendieck inequality once to the whole $n$-copy space. Its factor of two then costs only $1/(2n)$ bits per input copy. Applying the inequality separately on each copy would lose a factor $2^n$ and would not prove~\eqref{Grothendieck_asymptotic}.

Thus, within the class~\eqref{full_operator_hypothesis}, alternative quadratic forms cannot beat the regularised one-state construction on Werner inputs. This includes the two-state forms in Corollary~\ref{two_state_forms_cor}, since its Schatten-H\"older estimate also applies to arbitrary operators. It does not imply that the one-copy optimisation is optimal, and it says nothing analogous about arbitrary conditional bounds. Most importantly, it does not cover forms whose estimates use $0 \leq Q \leq \id$ essentially. The Rains derivation in~\eqref{rains_quadratic_estimate} is already of this latter kind: its use of $Q^2 \leq Q$ cannot be extended to arbitrary operators. Such effect-dependent forms remain a possible source of stronger converses.

\subsection{Converse via Hirschman's extension of the Hadamard three-line theorem}

We now turn to a different converse, based on analytic families of trace-one matrices. Lifting an annulus to a strip allows us to apply Hirschman's refinement of the Hadamard three-line theorem with explicit boundary weights. We will use the following scalar version of~\cite[Lemma~3.2, Eqs.~(26)--(29)]{Junge2018}; see also its Lemma~A.1. The scalar form suffices because the interpoland $h_n$ below is a scalar trace pairing.

\begin{lemma}[(Hirschman's three-line inequality)]\label{hirschman_lemma}
Let $h$ be bounded on $\{w\in\mathbb C:0 \leq \Re w \leq 1\}$, holomorphic in its interior and continuous on its boundary. For $0<\theta<1$, let $\beta_\theta$ be defined by~\eqref{hirschman_density}. These are nonnegative probability densities, and
\bb
\log|h(\theta)| \leq \int_{\mathbb R} \Big[(1 - \theta)\beta_{1 - \theta}(s)\log|h(is)| + \theta\beta_\theta(s)\log|h(1 + is)|\Big]\,\mathrm ds\, .
\label{hirschman_scalar}
\ee
Here $\log0 = - \infty$, with the inequality interpreted by regularisation at zeros. In particular it applies to a bounded periodic function on the strip.
\end{lemma}
In our application, the boundary logarithms admit finite continuous upper bounds. To handle possible zeros, it suffices to apply the operator-norm version of the cited lemma to $\operatorname{diag}(h,\varepsilon)$ and let $\varepsilon\to 0^+$. This replaces each logarithm by $\log\max\{|h|,\varepsilon\}$ and requires no subharmonic maximum principle.

We also recall the exact strong-converse exponent of quantum hypothesis testing. The sandwiched R\'enyi relative entropy is defined in~\eqref{sandwiched_renyi}.

\begin{lemma}[(Mosonyi and Ogawa's strong converse exponent~{\cite{Mosonyi2015}})]\label{mo_exponent_lemma}
Let $\rho,\sigma$ be finite-dimensional states with $\supp\rho\subseteq\supp\sigma$, and let $r \geq 0$. Among sequences of effects satisfying
\bb
\liminf_{n\to\infty} - \frac1n\log \Tr\!\left[\sigma^{\otimes n}T_n\right] \geq r\, ,
\label{testing_type_two_exponent}
\ee
the smallest possible lower exponential decay rate of $\Tr\!\left[\rho^{\otimes n}T_n\right]$ is
\bb
H_r^*(\rho\|\sigma) \coloneqq \sup_{\alpha>1}\frac{\alpha - 1}{\alpha} \bigl(r - \widetilde D_\alpha(\rho\|\sigma)\bigr) .
\label{mo_exact_exponent}
\ee
This rate can be attained. Moreover, every $n$, every effect $0 \leq T \leq \id$ and every $\alpha>1$ satisfy
\bb
\Tr\!\left[\rho^{\otimes n}T\right] \leq 2^{\frac{n(\alpha - 1)}{\alpha}\widetilde D_\alpha(\rho\|\sigma)} \left(\Tr\!\left[\sigma^{\otimes n}T\right]\right)^{\frac{\alpha - 1}{\alpha}}\, .
\label{mo_finite}
\ee
\end{lemma}
The exact exponent is given by~\cite[Theorem~IV.10 and Remark~IV.11]{Mosonyi2015}; the finite bound is~\cite[Eq.~(75), proof of Lemma~IV.7]{Mosonyi2015}. That paper uses natural exponential rates; all quantities here are converted to bits. The case $\rho = \sigma$ is elementary: both testing probabilities coincide, and $T_n = 2^{-nr}\id$ attains the exponent $r$. If the second probability in \eqref{mo_finite} vanishes, support inclusion makes the first one zero as well. The finite bound follows from binary-measurement data processing for \eqref{sandwiched_renyi}, retaining just the accepting term. This precise exponent strengthens the earlier Ogawa--Nagaoka estimate~\cite{Ogawa2000}.

\begin{manualthm}{\ref{analytic_converse_thm}}[(Analytic-family strong converse)]
Let $\rho = \rho_{A:B}$ be a finite-dimensional bipartite state, let $\mathcal{O} = \{z\in\mathbb C:a<|z - c|<b\}$ with $c\in\mathbb C$ and $0<a<b$, and let $Z(z)$ be an operator on $AB$, analytic in a neighbourhood of $\overline{\mathcal{O}}$, with $\Tr[Z(z)] = 1$. Suppose that $z_*\in\mathcal{O}$, $\sigma = Z(z_*)$ is a state, and $\supp\rho\subseteq\supp\sigma$. With the density in~\eqref{explicit_annular_density}, built from~\eqref{hirschman_density}, define
\begin{equation}
\begin{aligned}
J_{Z,\mathcal{O},z_*}(R) \coloneqq \int_{\partial\mathcal{O}} \min\!\left\{\log\|Z(z)\|_1,\,\log\big\|Z(z)^{\pt_B}\big\|_1 - R\right\}k_{\mathcal{O},z_*}(z)\,|\mathrm dz|\, .
\end{aligned}
\tag{\ref{analytic_J}}
\end{equation}
If $R \geq 0$ and $D(\rho\|\sigma) + J_{Z,\mathcal{O},z_*}(R)<0$, then $R$ is an exponential strong-converse rate for PPT distillation of $\rho$. Specifically,
\begin{equation}
\begin{aligned}
\EpptErr(\rho,R) \geq \sup_{\alpha>1}\frac{\alpha - 1}{\alpha}\Big[ - J_{Z,\mathcal{O},z_*}(R) - \widetilde D_\alpha(\rho\|\sigma)\Big]>0\, .
\end{aligned}
\tag{\ref{analytic_error_exponent}}
\end{equation}
\end{manualthm}
\begin{proof}
Write $z_* = c + r e^{i\varphi}$ and set
\bb
\ell \coloneqq \ln(b/a),\qquad \theta \coloneqq \frac{\ln(r/a)}{\ell}\in(0,1),\qquad f_{Z,R}(z) \coloneqq \min\!\left\{\log\|Z(z)\|_1,\log\big\|Z(z)^{\pt}\big\|_1 - R\right\} .
\label{analytic_strip_data}
\ee
The boundary integral in~\eqref{analytic_J} has the equivalent strip representation
\bb
J_{Z,\mathcal{O},z_*}(R) = \int_{\mathbb R}\Big[&(1 - \theta)\beta_{1 - \theta}(s)f_{Z,R}\!\left(c + a e^{i(\varphi + \ell s)}\right)\\
&+\theta\beta_\theta(s)f_{Z,R}\!\left(c + b e^{i(\varphi + \ell s)}\right)\Big]\,\mathrm ds\, .
\label{analytic_J_strip}
\ee
Appendix~\ref{harmonic_primer_app} verifies this equality and gives the explicit positive boundary densities. For these fixed data write $J(R) = J_{Z,\mathcal{O},z_*}(R)$.

Fix $n$ and a feasible effect $Q_n$ for output rank $d_n = \lceil2^{nR}\rceil$. By \eqref{effect}, $\big\|Q_n^{\pt}\big\|_\infty \leq 2^{-nR}$. Lift the annulus to the strip by
\bb
h_n(w) \coloneqq \Tr\!\left[Z\!\left(c + a e^{\ell w + i\varphi}\right)^{\otimes n}Q_n\right] .
\label{analytic_strip_lift}
\ee
This function is holomorphic near the closed strip and bounded there, since the matrix argument remains in a compact annulus. It is also periodic in the imaginary direction, and $h_n(\theta) = \Tr\!\left[\sigma^{\otimes n}Q_n\right]$. For any $z\in\overline{\mathcal{O}}$, trace-norm duality gives
\bb
\left|\Tr\!\left[Z(z)^{\otimes n}Q_n\right]\right| &\leqt{(i)} \|Z(z)\|_1^n\, ,\\
\left|\Tr\!\left[Z(z)^{\otimes n}Q_n\right]\right| &\eqt{(ii)} \left|\Tr\!\left[(Z(z)^{\pt})^{\otimes n}Q_n^{\pt}\right]\right| \leqt{(iii)} 2^{-nR}\big\|Z(z)^{\pt}\big\|_1^n\, .
\label{analytic_boundary_estimates}
\ee
Here (i) uses $\|Q_n\|_\infty \leq 1$ and tensor multiplicativity, (ii) uses invariance of the bilinear trace pairing under partial transpose, and (iii) uses feasibility and the same norm duality. No Hermiticity of $Z(z)$ is needed. Both boundary trace norms are at least one because their traces are one; hence $f_{Z,R}$ is finite and continuous.

Apply Lemma~\ref{hirschman_lemma} to $h_n$. On each boundary line the two estimates in \eqref{analytic_boundary_estimates} bound its logarithm by the corresponding $n f_{Z,R}$. Since the kernels are nonnegative, this gives
\bb
\Tr\!\left[\sigma^{\otimes n}Q_n\right] \leq 2^{nJ(R)}\, .
\label{analytic_reference_acceptance}
\ee
If a boundary zero is present, the regularisation convention following Lemma~\ref{hirschman_lemma}, applied to $h_n$, gives the same conclusion.

For every $\alpha>1$, Lemma~\ref{mo_exponent_lemma} now yields
\bb
\Tr\!\left[\rho^{\otimes n}Q_n\right] &\leqt{(i)} 2^{\frac{n(\alpha - 1)}{\alpha}\widetilde D_\alpha(\rho\|\sigma)} \left(\Tr\!\left[\sigma^{\otimes n}Q_n\right]\right)^{\frac{\alpha - 1}{\alpha}}\\
&\leqt{(ii)} 2^{-\frac{n(\alpha - 1)}{\alpha} [ - J(R) - \widetilde D_\alpha(\rho\|\sigma)]}\, .
\label{analytic_success_exponent}
\ee
Here (i) is \eqref{mo_finite} and (ii) is \eqref{analytic_reference_acceptance}. The estimate is uniform over all feasible effects. Take their supremum, then the lower exponential rate, and finally the supremum over fixed $\alpha>1$. This proves \eqref{analytic_error_exponent}. Its strict positivity follows from $\widetilde D_\alpha(\rho\|\sigma)\to D(\rho\|\sigma)$ as $\alpha\to 1^+$ and the strict hypothesis.

The same exponent bound applies to any integer-rank sequence with $\log d_n/n\to R$. Indeed the preceding finite bound then contains $J(\log d_n/n)$. The minimum defining $f_{Z,R}$ is $1$-Lipschitz in $R$, and the two kernel masses sum to one, so $J$ is $1$-Lipschitz. Passing to the limit for each fixed $\alpha$ and only then taking the supremum gives the same result. The exactness in Lemma~\ref{mo_exponent_lemma} concerns unconstrained hypothesis testing; no exact PPT exponent is asserted here.
\end{proof}

Optimising Theorem~\ref{analytic_converse_thm} over the analytic family, annulus and reference point gives
\bb
E_A(\rho) \coloneqq \inf_{Z,\mathcal{O},z_*}\inf\left\{ R \geq 0:\ D(\rho\|Z(z_*)) + J_{Z,\mathcal{O},z_*}(R)<0 \right\} .
\label{analytic_bound}
\ee
The outer infimum is over the data satisfying the hypotheses of Theorem~\ref{analytic_converse_thm}, and $J_{Z,\mathcal{O},z_*}$ is defined in~\eqref{analytic_J}. We use the convention $\inf\varnothing = +\infty$.
\begin{cor}[(Combined general converse)]\label{combined_cor}
With $E_Q$ and $E_A$ defined by~\eqref{quadratic_bound} and~\eqref{analytic_bound}, respectively, every finite-dimensional bipartite state satisfies
\bb
\Eppt(\rho) \leq \min\{R^\infty(\rho),E_Q(\rho),E_A(\rho)\}\, .
\label{combined}
\ee
\end{cor}
Here $R^\infty$ is defined in~\eqref{rains_definition}. Equation~\eqref{combined} collects the bounds established or used here; it is not an optimality claim among all possible converses. The first entry follows from Rains's bound~\cite[Theorem~4.13]{Rains2001}, after normalising the auxiliary operator and applying the bound to blocks. The other entries follow from Theorems~\ref{quadratic_converse_thm} and~\ref{analytic_converse_thm}, using~\eqref{quadratic_bound} and~\eqref{analytic_bound}. Combining~\eqref{Rains_bound} with~\eqref{combined} yields~\eqref{LOCC_converse_bounds}.

\section{Werner states: explicit bounds and certified separation}

We use the Werner parametrisation~\eqref{werner}, with $P_0$ and $P_1$ the symmetric and antisymmetric projections. The PPT region is $p \leq 1/2$.

\subsection{Achievability bounds}

Although Theorem~\ref{regularised_formula_thm} is stated in terms of $L$, the auxiliary functional $\widetilde{L}$ retains the absolute operators before taking their norm and satisfies $L\leq\widetilde{L}\leq\Eppt$. It can therefore give a sharper bound on a fixed number of copies. We use it for the explicit one-copy optimisation below, and use $L$ for the collective linear program; their regularisations agree.

The effect $A = P_1 + uP_0$, with $u = (d - 1)/(d + 1)$, already appears in Matthews and Winter's local discrimination protocol~\cite[Eq.~(13)]{VV-dh-Chernoff}. Here $A^{\pt} = q(\id - \Phi_d)$, where $q = d/(d + 1)$, and $\Tr\!\left[\rho_p A\right] = t = (d - 1 + 2p)/(d + 1)$. Proposition~\ref{FG_properties_prop}(iv) therefore gives $\widetilde{L}(\rho_p)$ at least the displayed binary divergence, and Theorem~\ref{regularised_formula_thm} gives
\bb
\Eppt(\rho_p) \geq \Dbin\!\left(\frac{d - 1 + 2p}{d + 1}\,\middle\|\,\frac d{d + 1}\right), \qquad \frac12<p \leq 1\, .
\label{werner_lower}
\ee

The measurement itself is known; \eqref{werner_lower} is its threshold-based distillation consequence. For fixed $d$ and $p = 1/2 + \delta$, the right side is $2\delta^2/(d\ln2) + O(\delta^3)$.

For $0 \leq k \leq n$, define the threshold effect
\bb
Q_{n,k} = \sum_{x\in\{0,1\}^n:\,|x| \geq k} \bigotimes_{j = 1}^n A^{x_j}(\id - A)^{1 - x_j}\, .
\label{threshold_effect}
\ee
This construction explicitly handles the signed Werner transform. On the sector with $i$ antisymmetric factors, the coefficient of \eqref{threshold_effect} is
\bb
f_{n,k}(i) = \Pr\{i + \operatorname{Bin}(n - i,u) \geq k\} = \sum_{j = \max\{0,k - i\}}^{n - i}\binom{n - i}{j}u^j(1 - u)^{n - i - j}\, .
\ee
On a partially transposed sector containing $\ell$ factors in $\Phi_d$, its eigenvalue is exactly $\Pr\{\operatorname{Bin}(n - \ell,q) \geq k\}$. All these eigenvalues are nonnegative and bounded by the $\ell = 0$ value, which equals the operator norm. Thus no absolute-value estimate of an alternating coefficient sum is needed.

We can in fact optimise the absolute-operator bound over all one-copy POVMs. Keep $d$ and $p$ fixed, and retain the notation $u = (d - 1)/(d + 1)$, $q = d/(d + 1)$ and $t = p + u(1 - p)$. Define
\bb
c_0(\mu) = \frac{d + 1 - d\mu}{2},\qquad c_1(\mu) = \frac{d - 1 - (d - 2)\mu}{2}\, ,
\ee
\bb
J_0(\mu)& = (1 - p)\log\frac{1 - p}{c_0(\mu)} + p\log\frac p{c_1(\mu)}\, ,\\
J_1(\mu)& = (1 - u)(1 - p)\log\frac{1 - p}{c_0(\mu)} + t\log\frac t{\mu q}\, .
\label{onecopy_functions}
\ee
Here $u,q,c_0,c_1$ depend on the fixed dimension $d$, while $t,J_0,J_1$ also depend on the fixed parameter $p$; only these fixed arguments are suppressed. The displayed variable of $c_0,c_1,J_0,J_1$ is $\mu$. Zero numerators contribute zero, and $J_1(0) = +\infty$.

\begin{prop}[(Exact one-copy POVM bound)]\label{onecopy_prop}
For every integer $d \geq 2$ and $p\in[1/2,1]$, there is a unique $\mu_*\in(0,1]$ with $J_0(\mu_*) = J_1(\mu_*)$, and
\bb
\Eppt(\rho_p) \geq \widetilde{L}(\rho_p) = \sup_{\MM} \widetilde{\ell}(\rho_p,\MM) = J_0(\mu_*) = J_1(\mu_*)\, .
\label{onecopy_value}
\ee
Here the supremum ranges over finite one-copy measurement channels. The optimum is attained by $\{(1 - u\lambda_*)P_0,(1 - \lambda_*)P_1,\lambda_*(P_1 + uP_0)\}$, where $\lambda_* = J_0'(\mu_*)/[J_0'(\mu_*) - J_1'(\mu_*)]$. At $p = 1/2$, $\mu_* = \lambda_* = 1$ and $\widetilde{L}(\rho_{1/2}) = 0$; at $p = 1$, $\mu_* = (d + 1)/(d + 2)$ and $\widetilde{L}(\rho_1) = \log((d + 2)/d)$. This is a one-copy functional optimum, not an operational converse.
\end{prop}
\begin{proof}
Let $\mathcal{W}$ be the Werner twirl and $\mathcal{T} = \Gamma\circ\mathcal{W}\circ\Gamma$. On a Hermitian $H$, writing $P = \id - \Phi_d$,
\bb
\mathcal{T}(H) = \frac{\Tr\!\left[PH\right]}{d^2 - 1}P + \Tr\!\left[\Phi_dH\right]\Phi_d, \qquad |\mathcal{T}(H)| \leq \mathcal{T}(|H|)\, .
\ee
The first identity is the two-block twirl formula; the second follows from $|\Tr\!\left[RH\right]| \leq \Tr\!\left[R|H|\right]$ for each positive block projector $R$. Thus twirling every POVM effect preserves its probability and cannot increase the weighted absolute-operator norm. This uses the commutative range of this twirl, not operator convexity of absolute value.

For an invariant effect $E = aP_0 + bP_1$, its two partial-transpose eigenvalues are $(a + b)/2$ and $((d + 1)a - (d - 1)b)/2$. Put $A = P_1 + uP_0$. If $a \geq ub$, refine $E = bA + (a - ub)P_0$; otherwise refine $E = (a/u)A + (b - a/u)P_1$. In each case the summands have the same sign on each partial-transpose block, so their absolute values add exactly. The log-sum inequality in the dual \eqref{povm_dual} shows that refinement cannot decrease $\widetilde{\ell}$. Merging common rays leaves the measurement channel $\MM_\lambda$ represented by
\bb
\pazocal{E}_\lambda = \big((1 - u\lambda)P_0,(1 - \lambda)P_1,\lambda A\big),\qquad0 \leq \lambda \leq 1\, .
\ee
With $\mu = \Tr\!\left[\sigma P\right]$, the three dual expectations for the unscaled rays are $c_0(\mu),c_1(\mu),\mu q$. Cancelling the ray coefficients inside the logarithms gives
\bb
\widetilde{\ell}(\rho_p,\MM_\lambda) = \min_{0 \leq \mu \leq 1}[(1 - \lambda)J_0(\mu) + \lambda J_1(\mu)]\, .
\ee
An absent ray is omitted, so at $\lambda = 0,\mu = 0$ the value is $J_0(0)$, with no $0\cdot\infty$ expression. Each objective is convex in $\mu$.

The functions in \eqref{onecopy_functions} have the natural-unit derivatives
\bb
(\ln2)J_0' = \frac{(1 - p)d}{2c_0} + \frac{p(d - 2)}{2c_1} \geq 0, \qquad (\ln2)J_1' = \frac{q(1 - p)}{c_0} - \frac t\mu\, .
\ee
The second derivative of $J_1$ is positive, and $(\ln2)J_1'(1) = 1 - 2p \leq 0$, so $J_1$ strictly decreases on $(0,1)$. Hence $J_0 - J_1$ strictly increases from $- \infty$ at zero. At one its nonnegativity is exactly log-sum for numerator pair $(u(1 - p),p)$ and denominator pair $(u/2,1/2)$; equality holds only for $p = 1/2$. This proves the unique crossing and its location.

Every convex objective above equals $V = J_0(\mu_*) = J_1(\mu_*)$ at the crossing, so its minimum is at most $V$. For $p>1/2$, the stated $\lambda_*$ lies in $[0,1)$ and makes its derivative vanish at $\mu_*$; convexity gives a global minimum there and proves equality. At $p = 1/2$, use $\lambda_* = 1$. At $p = 1$, solving $c_1(\mu) = \mu q$ gives the stated value and $\lambda_* = (d - 2)(d + 1)/[(d + 2)(d - 1)]$; for $d = 2$ this is zero and $J_0$ is constant. Attainment by the POVM functional allows its positive weights to approach zero on zero-probability outcomes, exactly as in Lemma~\ref{povm_dual_lemma}.
\end{proof}

Collective positive-partial-transpose effects improve the near-threshold coefficient of \eqref{werner_lower}. A certified $24$-copy qutrit effect gives an achievable curve $\kappa_{24}\delta^2 + O(\delta^3)$ at $p = 1/2 + \delta$ as $\delta\to 0^+$, with $\kappa_{24} = 1.290480005840672317\ldots>2/(3\ln2)$. This is a fixed achievable-curve expansion.

Proposition~\ref{FG_properties_prop}(iv) also applies to collective block effects, and Theorem~\ref{regularised_formula_thm} turns the resulting $\widetilde{L}$ bound into an operational rate. For $d = 3$, an explicit rational $24$-copy effect in the companion certificate obeys $0 \leq A \leq \id$, $0 \leq A^{\pt} \leq (2/5)\id$. Its polynomial acceptance $a(p)$ has $a(1/2) = 2/5$ and
\bb
a'(1/2) = \frac{4118621179428546867727313}{1283031261738279119093760} \eqqcolon s\, .
\ee
Consequently its achievable curve, $\Dbin(a(p)\|2/5)/24$ when $a(p)>2/5$ and zero otherwise, has the fixed-block expansion $\kappa_{24}\delta^2 + O(\delta^3)$ as $\delta\to 0^+$, where
\bb
\kappa_{24} = \frac{s^2}{48(2/5)(3/5)\ln2} = 1.290480005840672317\ldots>\frac2{3\ln2}\, .
\ee
This strictly improves the quadratic coefficient of \eqref{werner_lower} on a sufficiently small right neighbourhood of $p = 1/2$. It is an expansion of one achievable curve, not of the unknown operational rate. The certificate checks all $25$ physical coefficients and all $25$ partial-transpose rows exactly.

The hierarchy \eqref{regularisation} has an exact finite linear-program form for Werner states. For fixed $d,p,m$, let $r_i = \binom mi p^i(1 - p)^{m - i}$ be the physical type probabilities; the vector $r = (r_0,\ldots,r_m)$ depends on $p,m$. Define the dimension-dependent type transform
\bb
(C_m)_{ki} = 2^{-m}\sum_{\ell}(-1)^\ell(d - 1)^\ell(d + 1)^{k - \ell} \binom{k}{\ell}\binom{m - k}{i - \ell},\qquad 0 \leq i,k \leq m\, ,
\label{transform}
\ee
where out-of-range binomial coefficients are zero. If $\mathcal{V}_m$ is the set of all nonzero vertices of $\{v \geq 0:|C_mv| \leq \boldsymbol1\}$, then
\bb
\frac{L(\rho_p^{\otimes m})}{m} = \max_{\substack{\lambda_v \geq 0\\
\sum_v\lambda_vv = \boldsymbol1}} \frac1m\sum_{v\in\mathcal{V}_m}\lambda_v(r\cdot v)\log(r\cdot v)\, .
\label{LP}
\ee
An optimum uses at most $m + 1$ outcomes. Any exactly feasible subset of columns already yields an achievable rate. Full vertex enumeration is unnecessary for a lower certificate, while a claim of finite-level optimality requires it or a matching dual proof. Running maxima of the normalised levels converge to $\Eppt$; no speed of convergence is known.

\begin{proof}[Proof of \eqref{LP}]
Independent Werner twirls and copy permutations fix $\rho_p^{\otimes m}$ and do not increase $\big\|E_j^{\pt}\big\|_\infty$: their conjugates are mixtures of unitary conjugations. Hence they do not decrease the objective defining $L$. A symmetric effect is $E(e) = \sum_i e_i\Pi_i$, where $\Pi_i$ sums the orthogonal sectors with $i$ antisymmetric factors. Its probability is $r\cdot e$, and its partial-transpose norm is $\|C_me\|_\infty$. The formula \eqref{transform} follows by taking products of $P_0^{\pt} = (\id + d\Phi_d)/2$ and $P_1^{\pt} = (\id - d\Phi_d)/2$.

The transform is invertible: it is the restriction of the invertible tensor partial transpose between the two invariant permutation-symmetric operator spaces. Thus $\{v \geq 0:|C_mv| \leq \boldsymbol1\}$ is compact. Every nonzero vertex has $\|C_mv\|_\infty = 1$, since otherwise it admits a nontrivial radial perturbation. For $e\ne0$, put $c = \|C_me\|_\infty$ and decompose $e/c$ into vertices of its active signed-$C_m$ face. The zero vertex is absent. The refined effects have probabilities summing to the original probability and norms summing exactly to $c$. The log-sum inequality therefore cannot decrease the objective. Merging equal rays yields \eqref{LP}. Its compact feasible coefficient set has an optimum with at most $m + 1$ nonzero coefficients by the basic feasible-solution property of linear programming. This proof uses all vertices for the exact equality, while any feasible finite subset suffices for an achievable certificate.
\end{proof}

\subsection{Converses and a gap with the regularised Rains bound}

For comparison, the known additive Werner Rains expression is, for $d>2$,
\bb
R^\infty(\rho_p) = R(\rho_p) = \begin{cases} 0,&0 \leq p \leq \frac12\, ,\\
1 - h_2(p),&\frac12 \leq p \leq \frac12 + \frac1d\, ,\\
\log\frac{d - 2}{d} + p\log\frac{d + 2}{d - 2},&\frac12 + \frac1d \leq p \leq 1. \end{cases}
\label{rains}
\ee
The single-copy formula is given in~\cite[Corollary~5.7, Eq.~(5.46)]{Rains2001}; Werner additivity is established in Ref.~\cite[Theorem~7, Example~3 and Sec.~III.D.2]{BRT26}.

\begin{manualthm}{\ref{gap_Werner_thm}}[(Strict qutrit Werner separation)]
For $d = 3$ and $p = 25/26$,
\begin{equation}
\begin{aligned}
0.5836& \leq \Eppt(\rho_{25/26}) \leq 0.6211\\
&<R^\infty(\rho_{25/26}) = \frac{25}{26}\log5 - \log3 \approx 0.6477\, .
\end{aligned}
\tag{\ref{benchmark}}
\end{equation}
Each of $E_Q$, $U$ and $E_A$, defined in~\eqref{quadratic_bound}, \eqref{one_state_quadratic_bound} and~\eqref{analytic_bound}, respectively, is strictly below $R^\infty$ on this state.
\end{manualthm}
\begin{proof}
The lower endpoint is certified by an exactly feasible $48$-copy, $49$-outcome POVM in \eqref{LP}. For the quadratic separation, take the positive auxiliary $\tau = (2/13)\Phi_3 + (11/13)(\id - \Phi_3)/8$. Werner invariance makes the span of the symmetric and antisymmetric projectors a reducing subspace of $\LL_{\rho^{-1/2}}\circ\Gamma\circ\LL_\tau\circ\Gamma\circ\LL_{\rho^{-1/2}}$, whose logarithmic expectation is $- 2V_2[\rho,\Gamma\circ\LL_\tau\circ\Gamma]$. In the orthonormal order $(P_1/\sqrt3,P_0/\sqrt6)$ its restriction and the physical vector are
\bb
G = \begin{pmatrix}19/50& - 1/2\\
- 1/2&43/2\end{pmatrix} = D + \frac{aa^T}{50},\quad D = \diag(9/25,9),\quad a = \binom1{ - 25},\quad v = \frac1{\sqrt{26}}\binom51\, .
\ee
The resolvent identity gives
\bb
v^T(\ln G - \ln D)v = \int_0^\infty\frac{|a^T(D + t\id)^{-1}v|^2}{50 + a^T(D + t\id)^{-1}a}\,\mathrm dt>0\, .
\ee
The denominator is positive and $a^Tv\ne0$, so the numerator is not identically zero. Thus $U(\rho_{25/26}) \leq - v^T\log G\,v/2 < - v^T\log D\,v/2 = R^\infty(\rho_{25/26})$, and~\eqref{U_bounds} gives $E_Q(\rho_{25/26}) \leq U(\rho_{25/26})$. Rational interval enclosures for the eigenvalues of $G$ and the logarithms give
\bb
0.633093814438981533 < - \frac12v^T\log G\,v < 0.633093814438981534\, .
\label{quadratic_werner_certificate}
\ee
The analytic certificate in Appendix~\ref{certificate_app} gives the smaller explicit upper bound $E_A(\rho_{25/26}) \leq 0.62107$. This compares the two exhibited certificates; it does not assert an ordering between the optimised quantities $U$ and $E_A$.

\end{proof}
The analytic upper value follows from the rational M\"obius annulus detailed in Appendix~\ref{certificate_app}. Its negative residual is enclosed with directed rounding and a one-sided quadrature estimate; numerical optimisation only locates the parameters. The lower value is a feasible rate, not an asserted optimum of \eqref{LP}. Further exact POVMs give rates at least $.144222039901$ at $p = 3/4$ using $48$ copies, and $.414801754074$ at $p = 9/10$ using $32$ copies. All three lower certificates verify the physical coefficients, every signed partial-transpose row, the POVM identity and rational logarithm remainders; their $10^{-12}$-wide scalar enclosures do not give operational upper bounds.

The bounds can be propagated between parameters by a free-admixture rule: if $\tau$ is PPT, then
\bb
\Eppt(s\rho + (1 - s)\tau) \leq s\Eppt(\rho),\qquad 0 \leq s \leq 1\, .
\label{admixture}
\ee
Fix $0<s<1$ and $\delta>0$. Using $\lceil(s+\delta)n\rceil$ resource copies of $\rho$, sample a Bernoulli$(s)$ word of length $n$, place resource copies in its selected positions whenever their number does not exceed $\lceil(s+\delta)n\rceil$, and prepare $\tau$ elsewhere; on the complementary event prepare $\tau^{\otimes n}$. This deterministic PPT simulation has vanishing trace-norm error against $(s\rho+(1-s)\tau)^{\otimes n}$ by the Chernoff bound. Composing it with a protocol for the mixture gives rate divided by $s+\delta$, and then $\delta\to 0^+$ proves \eqref{admixture}, without a general rate-continuity assumption. Since $\rho_{1/2}$ is PPT, the ratio $\Eppt(\rho_p)/(p - 1/2)$ is nondecreasing for $p>1/2$. An upper certificate at $q$ therefore propagates to $p \leq q$ by the factor $(p - 1/2)/(q - 1/2)$, while a lower bound propagates in the opposite direction. Figure~\ref{werner_fig} uses these implications.

For $1/2<p \leq q \leq 1$, the identity
\bb
\rho_p = \frac{p - 1/2}{q - 1/2}\rho_q + \left(1 - \frac{p - 1/2}{q - 1/2}\right)\rho_{1/2},\qquad \rho_{1/2}^{\pt} = \frac{\id - \Phi_d}{d^2 - 1} \geq 0\, ,
\ee
proves the ratio monotonicity used in the figure. Therefore a valid upper bound $U_q$ gives $\Eppt(\rho_p) \leq (p - 1/2)U_q/(q - 1/2)$ to its left, including the zero endpoint, and a valid lower bound $L_p$ gives $\Eppt(\rho_q) \geq (q - 1/2)L_p/(p - 1/2)$ to its right. Pointwise minima of applicable upper bounds and maxima of applicable lower bounds remain valid. This rule does not assert general convexity or permit arbitrary interpolation between endpoints.

\section{Conclusion}

We have investigated entanglement distillation under PPT operations in the regime of asymptotically vanishing error. As explained in the Introduction, this setting removes the difficulty of characterising LOCC operations while retaining the other two features that make entanglement manipulation so challenging: non-zero error for any finite number of copies and asymptotic limit in the number of copies itself. Our regularised formula gives a concrete way of approaching the optimal distillation rate from below, through an optimisation over measurements. As a by-product, it yields a faithful lower bound in terms of the negativity, making the fact that every NPT state is distillable under PPT operations quantitatively precise.

On the converse side, we have introduced two constructions, based on quadratic forms and analytic families, that bound the PPT\,---\,and hence also the LOCC\,---\,distillable entanglement. Both improve on the regularised Rains bound already for two-qutrit Werner states. Thus, the Rains bound does not give the sought-after solution to PPT entanglement distillation, even for a family with very many symmetries. The explicit gap between the achievable and converse rates in Figure~\ref{werner_fig} also makes clear how much remains to be understood already in this simple example.

The main problem left open is to calculate the PPT distillable entanglement by means of a single-letter formula, or at least to decide whether it is computable at all for general states. Our regularised characterisation provides converging lower bounds, but no guarantee on their distance from the limit. A computability result would require a complementary way of certifying arbitrarily accurate upper bounds. Whether the converse techniques developed here can be made complete in this sense, or whether a different description of the asymptotic problem is needed, remains to be seen.

\bigskip
\noindent \textit{Acknowledgments.} I am grateful to Marco Parentin and especially to Bartosz Regula for countless illuminating discussions on the topic of PPT entanglement manipulation. I acknowledge financial support from the European Union (ERC StG ETQO, Grant Agreement no.\ 101165230).

\appendix
\section{Singular physical inputs}\label{quadratic_support_app}

We justify the support convention and the geometric-mean representation in~\eqref{V_geometric_derivative}. Throughout this appendix $\KK>0$ is fixed. For every $\omega \geq 0$,
\bb
\ker\big(\LL_{\sqrt\omega}\circ\KK^{-1}\circ\LL_{\sqrt\omega}\big) = \ker\LL_{\sqrt\omega}\, .
\ee
Since $\sqrt\omega = \LL_{\sqrt\omega}[\id]$ lies in the orthogonal complement of this kernel, its logarithmic expectation is finite.

The polar decomposition of $\KK^{-1/2}\circ\LL_{\sqrt\omega}$ gives
\bb
&\Braket{\sqrt\omega,\big(\log(\LL_{\sqrt\omega}\circ\KK^{-1}\circ\LL_{\sqrt\omega})\big)[\sqrt\omega]}\\
&\quad = \Braket{\KK^{1/2}[\id], \left[(\KK^{-1/2}\circ\LL_\omega\circ\KK^{-1/2}) \log(\KK^{-1/2}\circ\LL_\omega\circ\KK^{-1/2})\right][\KK^{1/2}[\id]]} ,
\label{continuous_quadratic_expression}
\ee
where $0\log0 = 0$. The right side is continuous in $\omega \geq 0$ and $\KK>0$, by continuous functional calculus. Moreover,
\bb
\Tr[(\LL_\omega\#_t\KK)[\id]] & = \Tr[(\KK\#_{1 - t}\LL_\omega)[\id]]\\
& = \Braket{\KK^{1/2}[\id], (\KK^{-1/2}\circ\LL_\omega\circ\KK^{-1/2})^{1 - t}[\KK^{1/2}[\id]]} .
\ee
The right derivative at $t = 0$ is minus $\ln2$ times~\eqref{continuous_quadratic_expression}. This proves~\eqref{V_geometric_derivative}, including its singular-input version. The representation is independent of the full-rank approximation used for $\omega$.

\section{Exponential strong converse with ancillary systems}\label{conditional_strong_converse_app}

We prove the strong-converse assertion in Theorem~\ref{quadratic_converse_thm}. The data-processing proof in the main text already gives the finite-error inequality~\eqref{conditional_finite_error}; the purpose of this appendix is only to obtain exponential decay of the success probability.

Fix an ancillary pair $\rho_{A'B'},\JJ$ as in that theorem and some $\delta>0$. The spectral construction below gives operators $X_n$ on $(A'B')^n$ such that
\bb
\Tr\rho_{A'B'}^{\otimes n}X_nX_n^\dag& \leq 1\, ,\\
\Tr\rho_{A'B'}^{\otimes n}(X_n - \id)(X_n - \id)^\dag& \leq 2^{-n c_\delta}\, ,\\
\Braket{X_n,\JJ^{\otimes n}[X_n]}& \leq 2^{-n\alpha(V_\alpha[\rho_{A'B'},\JJ] - \delta)}\, ,
\label{ancillary_truncation_estimates}
\ee
where $c_\delta>0$ is independent of $n$. Neither positivity nor Hermiticity of $X_n$ is required.

To construct them, consider the spectral projection of $(\JJ^{-1/2}\circ\LL_{\rho_{A'B'}}\circ\JJ^{-1/2})^{\otimes n}$ onto eigenvalues at least $2^{n\alpha(V_\alpha[\rho_{A'B'},\JJ] - \delta)}$, and denote it by $\Pi_n$. Set
\bb
X_n \coloneqq (\JJ^{-1/2})^{\otimes n} \Big[\Pi_n\big[(\JJ^{1/2}[\id])^{\otimes n}\big]\Big] .
\label{ancillary_truncation}
\ee
If $\sum_\lambda\lambda P_\lambda$ is the spectral decomposition of $\JJ^{-1/2}\circ\LL_{\rho_{A'B'}}\circ\JJ^{-1/2}$, the numbers $\lambda\big\|P_\lambda[\JJ^{1/2}[\id]]\big\|_2^2$, with $\lambda>0$, form a probability distribution. Indeed their sum is $\Tr\rho_{A'B'} = 1$. By~\eqref{continuous_quadratic_expression}, the expectation of $\log\lambda$ is $\alpha V_\alpha[\rho_{A'B'},\JJ]$. On tensor powers these logarithms add independently. Since the distribution has finite support, a scalar Chernoff bound gives an exponentially small probability below the threshold defining $\Pi_n$. This is the second line of~\eqref{ancillary_truncation_estimates}. The first line is the retained probability. The third follows by bounding the ordinary squared norm of the retained vector by its weighted squared norm divided by the spectral threshold. This proves all three assertions, also when $\rho_{A'B'}$ is singular.

Now let $Q_n$ be any principal-system test with $\big\|Q_n^{\pt}\big\|_\infty \leq 2^{-nR}$. Cauchy--Schwarz and the second estimate in~\eqref{ancillary_truncation_estimates} imply that
\bb
\left|\Tr\rho_{ABA'B'}^{\otimes n}(Q_n\otimes X_n) - \Tr\rho_{AB}^{\otimes n}Q_n\right| & \leq \left(\Tr\rho_{AB}^{\otimes n}Q_n^2\right)^{1/2} \left(\Tr\rho_{A'B'}^{\otimes n}(X_n - \id)(X_n - \id)^\dag\right)^{1/2}\\
& \leq 2^{-n c_\delta/2}\, .
\ee
On the other hand,
\bb
\Tr\rho_{ABA'B'}^{\otimes n}(Q_n^2\otimes X_nX_n^\dag)& \leq 1\, ,\\
\Braket{Q_n\otimes X_n,\KK^{\otimes n}[Q_n\otimes X_n]} & \leq \big\|Q_n^{\pt}\big\|_\infty^\alpha 2^{-n\alpha(V_\alpha[\rho_{A'B'},\JJ] - \delta)}\, .
\ee
The first inequality uses $Q_n^2 \leq \id$; the second is the tensor-stable ancillary hypothesis.

For any $\omega \geq 0$, $\KK>0$, arbitrary operator $Z$ and $0<s<1$, Cauchy--Schwarz followed by scalar Jensen gives
\bb
|\Tr\omega Z|^2 \leq \Tr[(\KK\#_{1 + s}\LL_\omega)[\id]] \big(\Tr\omega ZZ^\dag\big)^{1 - s}\Braket{Z,\KK[Z]}^s\, .
\label{quadratic_CS_general}
\ee
Indeed, apply Cauchy--Schwarz to
\bb
&(\KK^{-1/2}\circ\LL_\omega\circ\KK^{-1/2})^{(1 + s)/2}[\KK^{1/2}[\id]]\, ,\\
&(\KK^{-1/2}\circ\LL_\omega\circ\KK^{-1/2})^{(1 - s)/2}[\KK^{1/2}[Z]]\, ,
\ee
and then apply Jensen to the squared norm of the second vector, using the concavity of $x\mapsto x^{1 - s}$. The squared norm of the first vector is $\Tr[(\KK\#_{1 + s}\LL_\omega)[\id]]$. This calculation uses no inverse of $\omega$.

Apply~\eqref{quadratic_CS_general} to $\rho_{ABA'B'}^{\otimes n}$ and $Q_n\otimes X_n$. By the preceding estimates and tensorisation,
\bb
\Tr\rho_{AB}^{\otimes n}Q_n \leq &\;2^{-n c_\delta/2}\\
& + 2^{-\frac n2\left\{\alpha s\big(R + V_\alpha[\rho_{A'B'},\JJ] - \delta\big) - \log\Tr[(\KK\#_{1 + s}\LL_{\rho_{ABA'B'}})[\id]]\right\}}\, .
\label{conditional_exponential}
\ee
If $R>V_\alpha[\rho_{ABA'B'},\KK] - V_\alpha[\rho_{A'B'},\JJ]$, choose $\delta>0$ sufficiently small and then fix $s>0$ sufficiently small. Since
\bb
\lim_{s\to0^+}\frac1{\alpha s}\log\Tr[(\KK\#_{1 + s}\LL_{\rho_{ABA'B'}})[\id]] = V_\alpha[\rho_{ABA'B'},\KK]\, ,
\ee
both terms in~\eqref{conditional_exponential} decay exponentially. The constants are independent of the test. Taking the supremum over all feasible $Q_n$ proves the assertion.

For trivial ancillary systems no spectral truncation is needed. Applying~\eqref{quadratic_CS_general} directly to $\rho^{\otimes n}$ and $Q_n$ gives, for every $R \geq 0$ and $0<s<1$,
\bb
2\EpptErr(\rho,R) \geq \alpha sR - \log\Braket{\sqrt\rho, \big(\LL_{\sqrt\rho}\circ\KK^{-1}\circ\LL_{\sqrt\rho}\big)^s[\sqrt\rho]} .
\ee
The logarithmic term divided by $\alpha s$ tends to $V_\alpha[\rho,\KK]$ as $s\to0^+$. Moreover, scalar Jensen shows that the logarithmic expression divided by $s$ is non-decreasing in $s>0$, so leaving $s$ positive does not strengthen the asymptotic rate bound.

\section{Proof of the Grothendieck comparison}\label{Grothendieck_app}

We prove Proposition~\ref{Grothendieck_Werner_prop}. Fix $n$ and apply Theorem~\ref{little_Grothendieck_thm}, with $\HH = \HH' = \HH_{AB}^{\otimes n}$, to
\bb
T(X) = (\KK^{1/2})^{\otimes n}[X^{\pt}]\, .
\ee
By~\eqref{full_operator_hypothesis}, $\|T\|_{\infty\to2} \leq 1$. Hence there exist states $\tau_{0,n},\tau_{1,n}$ such that
\bb
\Braket{X,\KK^{\otimes n}[X]} \leq \Tr\tau_{0,n}(X^{\pt})^\dag X^{\pt} + \Tr\tau_{1,n}X^{\pt}(X^{\pt})^\dag
\label{Grothendieck_applied}
\ee
for every operator $X$.

Let $\PP_n$ be the Hilbert--Schmidt orthogonal projection onto the complex linear span of $P_{x_1}\otimes\cdots\otimes P_{x_n}$, with $x_j\in\{0,1\}$. This is the independent Werner twirl. If $X$ lies in its range, $X^{\pt}$ lies in the commutative algebra generated by the $n$ Bell projections, and is therefore normal, whether or not $X$ is Hermitian. Thus the two terms in~\eqref{Grothendieck_applied} combine. Twirl $(\tau_{0,n} + \tau_{1,n})/2$ independently by $U\otimes\overline U$ on each pair, and call the resulting state $\tau_n$. On the range of $\PP_n$,
\bb
\PP_n\circ\KK^{\otimes n}\circ\PP_n \leq 2\,\PP_n\circ\Gamma\circ\LL_{\tau_n}\circ\Gamma\circ\PP_n\, .
\label{Grothendieck_compressed}
\ee
The twirling does not change the quadratic expectations in this display.

The range of $\PP_n$ contains $\sqrt{\rho^{\otimes n}}$ and reduces both $\LL_{\rho^{\otimes n}}$ and $\Gamma\circ\LL_{\tau_n}\circ\Gamma$. For the moment, suppose $\tau_n>0$, and write
\bb
\GG \coloneqq \LL_{(\rho^{\otimes n})^{-1/2}}\circ\KK^{\otimes n}\circ\LL_{(\rho^{\otimes n})^{-1/2}}\, .
\ee
Operator Jensen for $\log$, applied to compression onto $\ran\PP_n$, followed by operator monotonicity, gives
\bb
n V_2[\rho,\KK] & = - \frac12\Braket{\sqrt{\rho^{\otimes n}},(\log\GG)[\sqrt{\rho^{\otimes n}}]}\\
& \geq - \frac12\Braket{\sqrt{\rho^{\otimes n}}, \left(\log\left(\left.\PP_n\circ\GG\circ\PP_n\right|_{\ran\PP_n}\right)\right) [\sqrt{\rho^{\otimes n}}]}\\
& \geq V_2[\rho^{\otimes n},\Gamma\circ\LL_{\tau_n}\circ\Gamma] - \frac12\\
& \geq U(\rho^{\otimes n}) - \frac12\, .
\label{Grothendieck_log_comparison}
\ee
The first inequality has this direction because $\log$ is operator concave. The loss $1/2$ in the next line is precisely $\frac12\log2$, coming from~\eqref{Grothendieck_compressed}.

If $\tau_n$ is singular, replace it with $(1 - \e)\tau_n + \e\id/d^{2n}$. The right side of~\eqref{Grothendieck_compressed} then acquires a factor $1/(1 - \e)$, so the same proof loses $\frac12\log(2/(1 - \e))$ instead. Taking $\e\to0^+$ proves~\eqref{Grothendieck_finite}. Finally, divide by $n$ and take the regularisation limit. The factor of two is incurred once, rather than once per copy, and therefore disappears. This proves~\eqref{Grothendieck_asymptotic}.

\section{Periodic boundary weights and the certificate kernels}
\label{harmonic_primer_app}

The strip representation \eqref{analytic_J_strip} is exactly the annular integral used by the existing certificates. We verify this identity directly from Lemma~\ref{hirschman_lemma}, without a Dirichlet problem. For the fixed $\theta\in(0,1)$ of \eqref{analytic_strip_data}, write
\bb
h_0(s) \coloneqq (1 - \theta)\beta_{1 - \theta}(s),\qquad h_1(s) \coloneqq \theta\beta_\theta(s)\, .
\label{strip_unnormalised_weights}
\ee
These even functions have masses $1 - \theta$ and $\theta$, respectively. For any real $u>0$, apply Hirschman's inequality to $\exp(e^{u w})$ and its reciprocal. Both are bounded on the closed strip, so the two inequalities force equality. Repeating with $- u$ in place of $u$ and cancelling the common conversion between natural and binary logarithms gives
\bb
e^{u\theta}& = \int_{\mathbb R} \big[h_0(s) + e^u h_1(s)\big]\cos(us)\,\mathrm ds\, ,\\
e^{-u\theta}& = \int_{\mathbb R} \big[h_0(s) + e^{-u}h_1(s)\big]\cos(us)\,\mathrm ds\, .
\label{kernel_two_linear_identities}
\ee
Solving these two linear equations yields
\bb
\int_{\mathbb R}h_0(s)\cos(us)\,\mathrm ds = \frac{\sinh((1 - \theta)u)}{\sinh u},\qquad \int_{\mathbb R}h_1(s)\cos(us)\,\mathrm ds = \frac{\sinh(\theta u)}{\sinh u}\, .
\label{strip_fourier_integrals}
\ee

The boundary values in \eqref{analytic_J_strip} are periodic with period $2\pi/\ell$ in $s$. Splitting the integral into periods therefore gives the nonnegative angular kernels, relative to $\mathrm d\psi/(2\pi)$,
\bb
k_j(\psi) \coloneqq \frac{2\pi}{\ell} \sum_{m\in\mathbb Z} h_j\!\left(\frac{\psi - \varphi + 2\pi m}{\ell}\right), \qquad j = 0,1\, .
\label{periodised_strip_kernels}
\ee
On the inner and outer circles, respectively, this gives
\bb
k_{\mathcal{O},z_*}(z) = \begin{cases}
\displaystyle\frac{k_0(\psi)}{2\pi a},&z = c + a e^{i\psi},\\
\displaystyle\frac{k_1(\psi)}{2\pi b},&z = c + b e^{i\psi}.
\end{cases}
\label{annulus_arclength_density}
\ee
The arclength element $|\mathrm dz|$ is positive on both circles. Hence no orientation is needed; using the usual positive contour orientation would require an extra minus sign on the clockwise inner circle.

The sums converge uniformly, because the strip kernels have exponential tails. Their masses are $1 - \theta$ and $\theta$. Their cosine coefficients follow from \eqref{strip_fourier_integrals} with $u = k\ell$. Since $r = a e^{\theta\ell}$, rewriting the hyperbolic sines gives
\bb
k_0(\psi)& = 1 - \theta + 2\sum_{k \geq 1} \frac{(a/r)^k\bigl[1 - (r/b)^{2k}\bigr]}{1 - (a/b)^{2k}} \cos\bigl(k(\psi - \varphi)\bigr) ,\\
k_1(\psi)& = \theta + 2\sum_{k \geq 1} \frac{(r/b)^k\bigl[1 - (a/r)^{2k}\bigr]}{1 - (a/b)^{2k}} \cos\bigl(k(\psi - \varphi)\bigr) .
\label{annulus_harmonic_representation}
\ee
The coefficients decay geometrically. Equality with \eqref{periodised_strip_kernels} follows first against trigonometric polynomials and then against all continuous functions by uniform approximation. At the certificate reference point $- r_0$ one has $\varphi = \pi$ and $r = r_0$. Thus the kernels, their masses and every certified integral are unchanged. This also identifies the two boundary weights with the harmonic measure of the annulus at the reference point.

\section{Certified annular parameters and quadrature}\label{certificate_app}

For $d = 3$, set $Z_w = (1 - w)P_0/6 + wP_1/3$. Orthogonal spectral sectors give
\bb
A(w) = \|Z_w\|_1 = |w| + |1 - w|,\qquad B(w) = \big\|Z_w^{\pt}\big\|_1 = \frac23\big(|w + 1| + |w - 1/2|\big) .
\ee
Use $w(\zeta) = c + \zeta/(1 + h\zeta)$ on $a<|\zeta|<b$, with the exact rational parameters
\bb
t = \frac{959483}{10^6},\quad c = \frac{987646}{10^6},\quad a = \frac{268002}{10^7},\quad b = \frac{370446}{10^6},\quad h = - \frac{1114048}{10^6}\, .
\ee
The reference point is $- r_0$, where $r_0 = (c - t)/(1 + h(c - t))$. Rational inequalities verify $|h|b<1$ and $0<a<r_0<b$, so the family is analytic near the closed annulus and $Z_{w( - r_0)} = \rho_t$ is full rank. Let $q_i = a/r_0$, $q_o = r_0/b$, $s = a/b$, $\lambda_i = \ln(b/r_0)/\ln(b/a)$ and $\lambda_o = 1 - \lambda_i$. The harmonic kernels, relative to $\mathrm d\theta/(2\pi)$, are
\bb
k_i(\theta) = \lambda_i + 2\sum_{k \geq 1} \frac{q_i^k(1 - q_o^{2k})}{1 - s^{2k}}\cos k(\theta - \pi),\qquad k_o(\theta) = \lambda_o + 2\sum_{k \geq 1} \frac{q_o^k(1 - q_i^{2k})}{1 - s^{2k}}\cos k(\theta - \pi)\, .
\ee
These follow by periodising Hirschman's weights as in Appendix~\ref{harmonic_primer_app}. They are nonnegative, with masses $\lambda_i$ and $\lambda_o$ whose sum is one. With $J$ defined by \eqref{analytic_J}, a directed-rounding certificate gives
\bb
\Dbin(25/26\|t) + J(6211/10000) \leq - 1.1269845955\times10^{-6}\, .
\ee
The integration uses $65\,536$ midpoint cells on $[0,\pi]$, six terms in a geometric rearrangement of each kernel, and 60-digit decimal intervals. The one-sided quadrature error is at most $4.499149143\times10^{-7}$ and the kernel-tail error below $3.088\times10^{-14}$. Branch crossings require no mesh alignment: if $k \geq 0$ and $|(kf_j)''| \leq M_2$ on a cell of width $\ell$, the cell average of $k\min(f_0,f_1)$ is at most its midpoint value plus $M_2\ell^2/24$, by taking the smaller of the two Taylor integral bounds. The supplied certificate verifies the branch derivative bounds and all outward-rounded elementary operations.

On the inner circle, exact geometry and rational logarithm bounds give $\log(B/A)>.67525>.6211$. Hence its minimum branch is $\log A$ and remains unchanged when the rate is decreased. The outer harmonic mass satisfies $\lambda_o<.0311$. Lowering the rate by $3/100000$ therefore gives the exact bound
\bb
\Dbin(25/26\|t) + J(.62107) \leq - \frac{112698}{10^{11}} + \frac3{100000}\frac{311}{10000} = - \frac{9699}{50000000000}<0\, .
\ee
Theorem~\ref{analytic_converse_thm} proves the asserted single-pair strong converse. The improvement from $.6211$ to $.62107$ uses this residual slack and the small outer harmonic mass; it does not require a claim that correlated auxiliary families outperform all one-pair families. Exact rational feasibility checks and interval outputs are retained with the companion research records. They certify the displayed parameters, not optimal analytic geometry or an interpolated curve.

\section{Companion certificate files}\label{supplement_app}

The separate computational supplement contains the rational feasibility witnesses, interval arithmetic and executable checks supporting the numerical bounds in this paper. Its \path{README.md} maps each claim to its inputs and verification code, while \path{certificate-manifest.json} records their SHA-256 hashes. The program \path{run_certificates.py} verifies the three strongest POVM records, the $24$-copy effect, the annular quadrature and residual-slack improvement, the corrected three-copy analytic certificate, and the quadratic comparison above. It requires only the Python standard library. The full replay and its outputs are supplied in \path{fresh-verification/}.

The figure data, construction code and independent assembly audit are also included, with the scope of their checks documented separately. The displayed curve samples are floating-point evaluations of proved bounds, not interval certificates. Historical proof records can use earlier equation numbers; the supplement's index identifies the corresponding mathematical claims.

\bibliography{../../biblio}
\end{document}

%% file: macros.tex
\usepackage{newpxtext,newpxmath}

\usepackage[latin1]{inputenc}
\usepackage{amsthm}
\usepackage{amssymb}
\usepackage{amsmath}
\usepackage{bbold}
\usepackage{bbm}
\usepackage[pdftex, backref=page]{hyperref}
\usepackage{braket}
\usepackage{dsfont}
\usepackage{mathdots}
\usepackage{mathtools}
\usepackage{enumerate}
\usepackage[shortlabels]{enumitem}
\usepackage{csquotes}
\usepackage{stmaryrd}
\usepackage[cal=boondox]{mathalfa}
\usepackage{graphicx}
\usepackage{stackengine}
\usepackage{scalerel}
\usepackage{tensor}       
\usepackage{array}
\usepackage{makecell}
\newcolumntype{x}[1]{>{\centering\arraybackslash}p{#1}}
\usepackage{tikz}
\usepackage{pgfplots}
\usetikzlibrary{shapes.geometric, shapes.misc, positioning, arrows, arrows.meta, decorations.pathreplacing, decorations.pathmorphing, patterns, angles, quotes, calc}
\usepackage{booktabs}
\usepackage{xfrac}
\usepackage{siunitx}
\usepackage{centernot}
\usepackage{comment}
\usepackage{chngcntr}
\usepackage[justification=justified,format=plain]{caption}
\usepackage{subcaption}
\usepackage{ragged2e}

\newtheorem{thm}{Theorem}
\newtheorem*{thm*}{Theorem}
\newtheorem{prop}[thm]{Proposition}
\newtheorem*{prop*}{Proposition}
\newtheorem{lemma}[thm]{Lemma}
\newtheorem*{lemma*}{Lemma}
\newtheorem{cor}[thm]{Corollary}
\newtheorem*{cor*}{Corollary}

\newtheorem*{cj*}{Conjecture}

\newtheorem*{Def*}{Definition}

\newtheorem*{question*}{Question}

\newtheorem*{problem*}{Problem}

\makeatletter
\def\thmhead@plain#1#2#3{%
  \thmname{#1}\thmnumber{\@ifnotempty{#1}{ }\@upn{#2}}%
  \thmnote{ {\the\thm@notefont#3}}}
\let\thmhead\thmhead@plain
\makeatother

\theoremstyle{definition}

\makeatletter
\newcommand{\manualifempty}[3]{%
  \edef\@tempa{#1}%
  \ifx\@tempa\@empty
    #2
  \else
    #3
  \fi
}
\makeatother

\makeatletter
\newtheoremstyle{manualstyle}
  {3pt}{3pt}{\itshape}{}{\bfseries}{.}{ }{}

\theoremstyle{manualstyle}

\newtheorem{manualthminner}{Theorem}
\newenvironment{manualthm}[1]{%
  \def\@tempa{#1}%
  \ifx\@tempa\@empty
    \renewcommand{\themanualthminner}{\unskip}%
  \else
    \renewcommand{\themanualthminner}{#1}%
  \fi
  \begin{manualthminner}
}{\end{manualthminner}}

\newtheorem{manualpropinner}{Proposition}

\newtheorem{manuallemmainner}{Lemma}

\newtheorem{manualcorinner}{Corollary}

\makeatother

\newcommand{\bb}{\begin{equation}\begin{aligned}\hspace{0pt}}
\newcommand{\bbb}{\begin{equation*}\begin{aligned}}
\newcommand{\ee}{\end{aligned}\end{equation}}
\newcommand{\eee}{\end{aligned}\end{equation*}}

\newcommand{\eqt}[1]{\stackrel{\mathclap{\scriptsize \mbox{#1}}}{=}}
\newcommand{\leqt}[1]{\stackrel{\mathclap{\scriptsize \mbox{#1}}}{\leq}}

\newcommand{\geqt}[1]{\stackrel{\mathclap{\scriptsize \mbox{#1}}}{\geq}}
\newcommand{\ketbra}[1]{\ket{#1}\!\!\bra{#1}}

\newcommand{\ketbraa}[2]{\ket{#1}\!\!\bra{#2}}

\newcommand{\sumno}{\sum\nolimits}

\newcommand{\e}{\varepsilon}
\renewcommand{\epsilon}{\varepsilon}

\newcommand{\id}{\mathds{1}}

\newcommand{\N}{\mathds{N}}

\newcommand{\C}{\mathds{C}}

\newcommand{\locc}{\mathrm{LOCC}}

\newcommand{\ppt}{\mathrm{PPT}}

\newcommand{\PPT}{\pazocal{P\!P\!T}}

\DeclareMathOperator{\Tr}{Tr}

\DeclareMathAlphabet{\pazocal}{OMS}{zplm}{m}{n}

\DeclareMathOperator{\supp}{supp}

\DeclareMathOperator{\Id}{id}

\DeclareMathOperator{\diag}{diag}
\DeclareMathOperator{\ran}{ran}

\newcommand{\HH}{\pazocal{H}}

\newcommand{\MM}{\pazocal{M}}

\newcommand{\PP}{\pazocal{P}}

\newcommand{\lsmatrix}{\left(\begin{smallmatrix}}
\newcommand{\rsmatrix}{\end{smallmatrix}\right)}

\newcommand{\rel}[3]{#1\big(#2\,\big\|\,#3\big)}

\stackMath

\stackMath

\makeatletter
\newcommand*\rel@kern[1]{\kern#1\dimexpr\macc@kerna}
\newcommand*\widebar[1]{%
  \begingroup
  \def\mathaccent##1##2{%
    \rel@kern{0.8}%
    \overline{\rel@kern{-0.8}\macc@nucleus\rel@kern{0.2}}%
    \rel@kern{-0.2}%
  }%
  \macc@depth\@ne
  \let\math@bgroup\@empty \let\math@egroup\macc@set@skewchar
  \mathsurround\z@ \frozen@everymath{\mathgroup\macc@group\relax}%
  \macc@set@skewchar\relax
  \let\mathaccentV\macc@nested@a
  \macc@nested@a\relax111{#1}%
  \endgroup
}

\counterwithin*{equation}{part}
\counterwithin*{thm}{part}
\counterwithin*{figure}{part}

\tikzset{meter/.append style={draw, inner sep=10, rectangle, font=\vphantom{A}, minimum width=30, line width=.8, path picture={\draw[black] ([shift={(.1,.3)}]path picture bounding box.south west) to[bend left=50] ([shift={(-.1,.3)}]path picture bounding box.south east);\draw[black,-latex] ([shift={(0,.1)}]path picture bounding box.south) -- ([shift={(.3,-.1)}]path picture bounding box.north);}}}
\tikzset{roundnode/.append style={circle, draw=black, fill=gray!20, thick, minimum size=10mm}}
\tikzset{squarenode/.style={rectangle, draw=black, fill=none, thick, minimum size=10mm}}

\definecolor{Blues5seq1}{RGB}{239,243,255}
\definecolor{Blues5seq2}{RGB}{189,215,231}
\definecolor{Blues5seq3}{RGB}{107,174,214}
\definecolor{Blues5seq4}{RGB}{49,130,189}
\definecolor{Blues5seq5}{RGB}{8,81,156}

\definecolor{Greens5seq1}{RGB}{237,248,233}
\definecolor{Greens5seq2}{RGB}{186,228,179}
\definecolor{Greens5seq3}{RGB}{116,196,118}
\definecolor{Greens5seq4}{RGB}{49,163,84}
\definecolor{Greens5seq5}{RGB}{0,109,44}

\definecolor{Reds5seq1}{RGB}{254,229,217}
\definecolor{Reds5seq2}{RGB}{252,174,145}
\definecolor{Reds5seq3}{RGB}{251,106,74}
\definecolor{Reds5seq4}{RGB}{222,45,38}
\definecolor{Reds5seq5}{RGB}{165,15,21}

\allowdisplaybreaks

\usepackage[most,breakable]{tcolorbox}
\makeatletter

\def\boxed@gobegin#1{\def\@tempa{#1}\def\@tempb{orange}\ifx\@tempa\@tempb\begin{tcolorbox}[colback=red!15,colframe=orange!70,breakable,enhanced]\else\begin{tcolorbox}[colback=Blues5seq1,colframe=Blues5seq5,breakable,enhanced]\fi}
\def\boxed@gobegin@empty{\begin{tcolorbox}[colback=Blues5seq1,colframe=Blues5seq5,breakable,enhanced]}
\makeatother